\documentclass[a4paper,accepted=2025-03-26,onecolumn]{quantumarticle}
\pdfoutput=1
\usepackage{fullpage}

\usepackage[
  colorlinks,
  linkcolor = blue,
  citecolor = blue,
  urlcolor = blue]{hyperref}

\usepackage[T1]{fontenc}

\usepackage{amsthm, amsmath, amsfonts, amssymb}
\usepackage{enumerate}
\usepackage[scr=rsfs,scrscaled=1]{mathalpha}
\usepackage{tikz}

\usepackage{subfigure}
\usepackage{dsfont}
\usepackage{mathrsfs}
\usepackage{microtype}

\usepackage{thmtools,thm-restate}

\theoremstyle{plain}
\newtheorem{theorem}{Theorem}[section]
\newtheorem{lemma}[theorem]{Lemma}
\newtheorem{corollary}[theorem]{Corollary}
\newtheorem{conjecture}[theorem]{Conjecture}
\newtheorem{definition}[theorem]{Definition}
\newtheorem{prop}[theorem]{Proposition}

\theoremstyle{definition}
\newtheorem{example}[theorem]{Example}

\usepackage[capitalise]{cleveref}
\crefname{theorem}{Theorem}{Theorems}
\crefname{lemma}{Lemma}{Lemmas}
\crefname{proposition}{Proposition}{Propositions}
\crefname{definition}{Definition}{Definitions}
\crefname{corollary}{Corollary}{Corollaries}
\crefname{example}{Example}{Examples}
\crefname{section}{Section}{Sections}
\crefname{appendix}{Appendix}{Appendices}
\crefname{table}{Table}{Tables}
\crefname{conjecture}{Conjecture}{Conjectures}

\DeclareMathOperator{\tr}{tr}
\DeclareMathOperator{\Enc}{Enc}
\DeclareMathOperator{\Dec}{Dec}

\DeclareMathOperator{\im}{Im}

\newcommand{\omegac}{\omega_{\mathrm{c}}}
\newcommand{\omegaq}{\omega_{\mathrm{q}}}
\newcommand{\omegans}{\omega_{\mathrm{ns}}}
\newcommand{\pg}[5][]{\omega^{#1}_{\mathrm{#2}}(#3|#4)_{#5}}

\newcommand{\s}{\mathsf}

\newcommand{\expectedbraket}[1]{\langle #1\rangle}

\usepackage{mathtools}
\DeclarePairedDelimiter\bra{\langle}{\rvert}
\DeclarePairedDelimiter\ket{\lvert}{\rangle}
\DeclarePairedDelimiterX\braket[2]{\langle}{\rangle}{#1 \delimsize\vert #2}

\newcommand{\ketbra}[2]{\ket{#1}\bra{#2}}
\newcommand{\proj}[1]{\ketbra{#1}{#1}}

\newcommand{\R}{\mathbb{R}}
\newcommand{\C}{\mathbb{C}}

\newcommand{\N}{\mathbb{N}}

\newcommand{\POVM}[1]{\mathrm{M}(#1)}
\newcommand{\PM}[1]{\mathrm{PM}(#1)}
\newcommand{\D}[1]{\mathrm{D}(#1)}
\newcommand{\U}[1]{\mathrm{U}(#1)}

\newcommand{\reg}[1]{\mathsf{#1}}
\newcommand{\A}{\reg{A}}
\newcommand{\B}{\reg{B}}
\newcommand{\X}{\reg{X}}

\newcommand{\AB}{\reg{AB}}

\newcommand{\XAB}{\reg{XAB}}
\renewcommand{\a}{\reg{A'}}
\renewcommand{\b}{\reg{B'}}
\newcommand{\ab}{{\a\b}}
\newcommand{\Aa}{{\A\a}}
\newcommand{\Bb}{{\B\b}}

\newcommand{\hilb}[1]{\mathcal{#1}}
\newcommand{\HH}{\hilb{H}}
\newcommand{\HA}{\hilb{A}}
\newcommand{\HB}{\hilb{B}}

\newcommand{\HX}{\hilb{X}}
\newcommand{\Ha}{\hilb{A'}}
\newcommand{\Hb}{\hilb{B'}}
\newcommand{\da}{d}
\newcommand{\db}{d}

\newcommand{\fset}[1]{\mathscr{#1}}
\newcommand{\SA}{\fset{A}}
\newcommand{\SB}{\fset{B}}

\newcommand{\SX}{\fset{X}}
\newcommand{\SY}{\fset{Y}}

\usepackage{mathtools}
\mathtoolsset{centercolon}
\DeclarePairedDelimiter{\set}{\lbrace}{\rbrace}
\DeclarePairedDelimiter{\abs}{\lvert}{\rvert}
\DeclarePairedDelimiter{\card}{\lvert}{\rvert}
\DeclarePairedDelimiter{\norm}{\lVert}{\rVert}
\DeclarePairedDelimiter{\of}{\lparen}{\rparen}
\DeclarePairedDelimiter{\pr}{\lparen}{\rparen}
\DeclarePairedDelimiter{\sof}{\lbrack}{\rbrack}

\newcommand{\x}{\otimes}
\newcommand{\tp}{^{\mathsf{T}}}
\newcommand{\ct}{^{\dagger}}

\newcommand{\Id}{\mathbb{I}} % identity matrix

\newcommand{\0}{\varnothing} % empty set
\newcommand{\eqdef}{:=}

\newcommand{\WolframRef}[1]{\href{https://reference.wolfram.com/language/ref/#1.html}{\texttt{#1}}}
\newcommand{\file}[1]{``\texttt{#1}''}

\usepackage{xcolor}
\newcommand{\edited}[1]{{#1}}
\newcommand{\off}[1]{} % {#1} % on / off

\usepackage{authblk}
\newcommand{\aff}[1]{\textit{\normalsize#1}}

\author[1,2]{Llorenç Escol\`a Farr\`as}
\author[3]{Jar\`on Has}
\author[1,3,4]{Maris Ozols}
\author[1,2]{Christian Schaffner}
\author[5]{Mehrdad Tahmasbi}

\affil[1]{\aff{\href{https://qusoft.org}{QuSoft}, Amsterdam, The Netherlands}}
\affil[2]{\aff{Informatics Institute, University of Amsterdam, The Netherlands}}
\affil[3]{\aff{Korteweg-de Vries Institute (KdVI), University of Amsterdam, The Netherlands}}
\affil[4]{\aff{Institute for Logic, Language, and Computation (ILLC), University of Amsterdam, The Netherlands}}
\affil[5]{\aff{Tufts University, Medford, MA, USA}}

\title{Parallel repetition of local simultaneous state\newline discrimination}

\begin{document}
\maketitle

\begin{abstract}
Local simultaneous state discrimination (LSSD) is a recently introduced problem in quantum information processing.
Its classical version is a non-local game played by non-commu\-ni\-ca\-ting players against a referee.
Based on a known probability distribution, the referee generates one input for each of the players and keeps one secret value. The players have to guess the referee's value and win if they all do so.
For this game, we investigate the advantage of no-signalling strategies over classical ones.
We show numerically that for three players and binary values, no-signalling strategies cannot provide any improvement over classical ones. For a certain LSSD game based on a binary symmetric channel, we show that no-signalling strategies are strictly better when multiple simultaneous instances of the game are played. Good classical strategies for this game can be defined by codes, and good no-signalling strategies by list-decoding schemes. We expand this example game to a class of games defined by an arbitrary channel, and extend the idea of using codes and list decoding to define strategies for multiple simultaneous instances of these games. Finally, we give an expression for the limit of the exponent of the classical winning probability, and show that no-signalling strategies based on list-decoding schemes achieve this limit.
\end{abstract}

\tableofcontents

\section{Introduction}

The task of discriminating between states is of fundamental importance in information processing and cryptography \cite{Blahut_1974, Polyanskiy2010, Maurer2000}.
A rich and extensive literature exists on this fundamental problem under the name of state discrimination or hypothesis testing \cite{Wasserman_2010, Helstrom_1969, Bae_Kwek_2015}.
In quantum cryptography and quantum information theory, a natural extension of state-discrimination problem is to distinguish \emph{quantum states}. In the context of non-local games, the state-discrimination problem arises in a multi-player setting. In these scenarios, it is interesting to study how non-local resources such as shared randomness, quantum entanglement or no-signaling correlations can help the players to succeed in the state-discrimination task.
Authors of \cite{Nonlocality, Framework} have studied the scenario where local operation and classical communication are allowed between two parties, and they have shown that entanglement can help the players. 

The authors of \cite{majenz2021local} studied another variant of distributed state discrimination in which multiple parties cannot communicate and have to estimate the state locally and simultaneously, hence calling the problem \emph{local simultaneous state discrimination} (LSSD). LSSD problems naturally arise in the context of uncloneable cryptography \cite{broadbent2019uncloneable, MST21, Prabhanjan2022, CLLZ21}, where we encode classical data into a quantum state such that an adversary cannot copy it. In such scenarios, successfully copying translates into successfully distinguishing quantum states. \edited{LSSD problems also appear in the study of monogamy of entanglement games \cite{Tomamichel_2013}, where two parties prepare a tripartite state and perform a measurement to guess the outcome of a measurement performed by a third party. Optimal performance of such games has been crucial to prove the security of uncloneable cryptographic schemes \cite{broadbent2019uncloneable}.} Depending on the resources shared between the parties, one can consider various strategies. The authors of \cite{majenz2021local} showed that even when the state has a classical description, quantum entanglement could enhance the probability of simultaneous state discrimination, and a more powerful resource of no-signaling correlations could enhance it even further.

As \cite{majenz2021local} have shown that finding the optimal strategy for three-party LSSD is NP-hard, it is likely to be challenging to study LSSDs in general. One could, however, characterize the optimal probability of winning and optimal strategies for LSSDs with some specific structure. One natural structure of interest is when an LSSD problem consists of several independent and identical LSSDs, and the parties have to win all these games at once in parallel. We call this type of LSSDs \emph{parallel repetition} of LSSDs, for which we establish several results in this article. Studying parallel LSSD games might have cryptographic implications. Many protocols have product structures, and if we restrict the adversaries only to applying a ``product'' attack, then the performance of such protocols is governed by parallel repetition of LSSDs. \edited{Furthermore, the monogamy of entanglement games with product structure have been important to understand. If we restrict the strategies to those with product states, then the problem can be formulated in terms of parallel repetition of LSSD games.} 

\subsection{Our contributions}

As a first simple observation, we show in \cref{the:symmetry} that for symmetric LSSD problems with classical inputs (as depicted in \cref{fig:LSSD_schematic}), there exists an optimal symmetric strategy. In other words, for an LSSD problem defined by a joint distribution $P_\s{XAB}$ such that $P_\s{X}$ is uniform over $\SX$, $P_\s{AB|X} = P_\s{A|X}P_\s{B|X}$, and $P_\s{A|X} = P_\s{B|X}$, there exist optimal classical deterministic strategies for Alice and Bob that are identical.

In \cref{sec:BSC_game} we analyze an example of an LSSD game introduced in~\cite{majenz2021local}, where the referee sends a bit $x$ over a \emph{binary symmetric channel} (BSC), see \cref{fig:BSC}, to Alice and Bob. We use the symmetry observation above to find optimal classical strategies for two and three parallel repetitions of this game in \cref{result w_c and w_ns repetition,result 3 copies NS}, respectively. We also give optimal no-signalling strategies for two and three copies (our results for two copies are depicted in \cref{fig:twoCopies}). Finally, in \cref{sec: arbitrary copies BSC}, we consider the $n$-fold parallel repetition of this game, and argue how the classical strategies relate to (regular) error-correcting codes and the no-signaling strategies relate to list-decoding schemes.

\begin{figure}[ht]
\centering
\subfigure[]{\includegraphics[width=78mm]{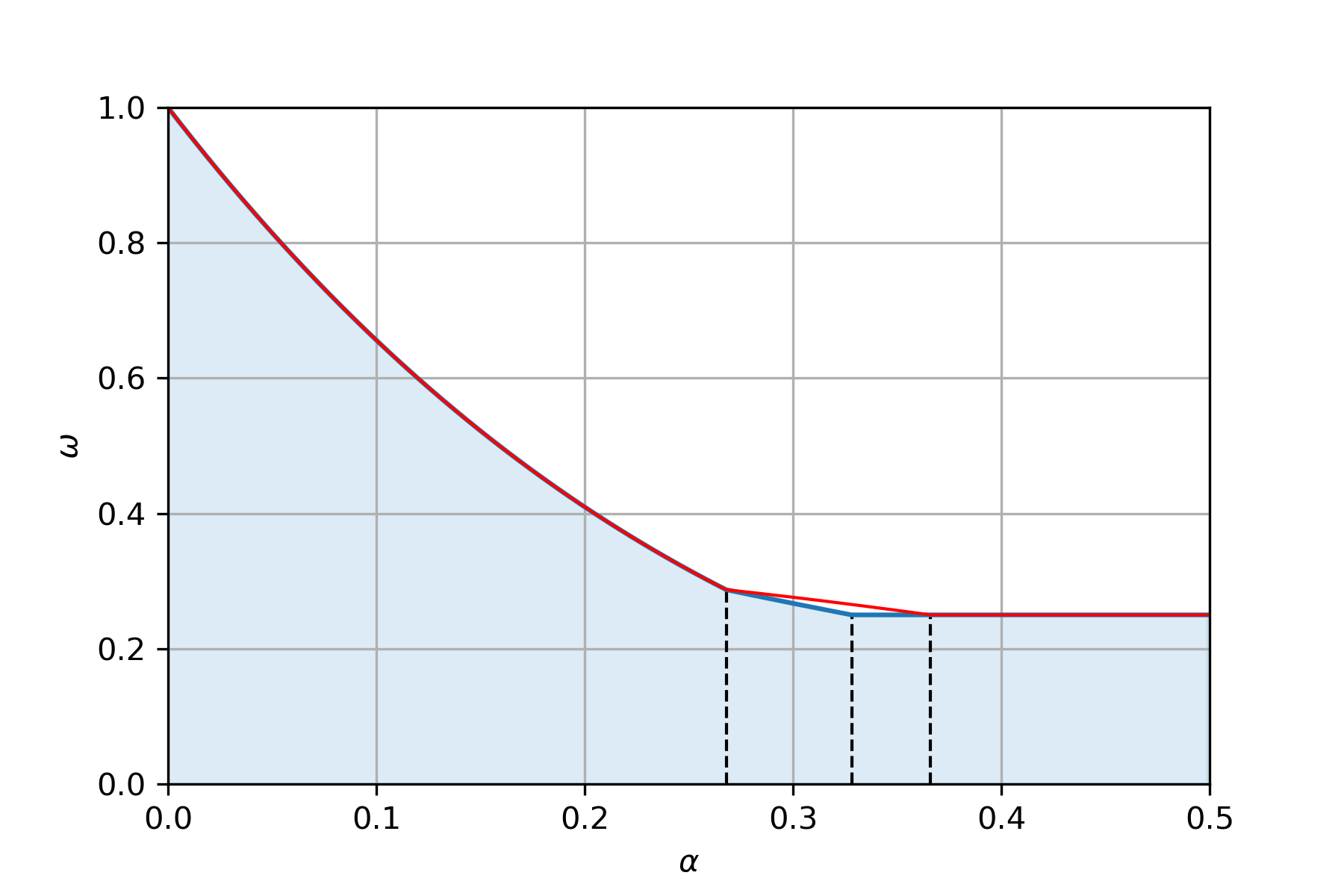}}
\subfigure[]{\includegraphics[width=78mm]{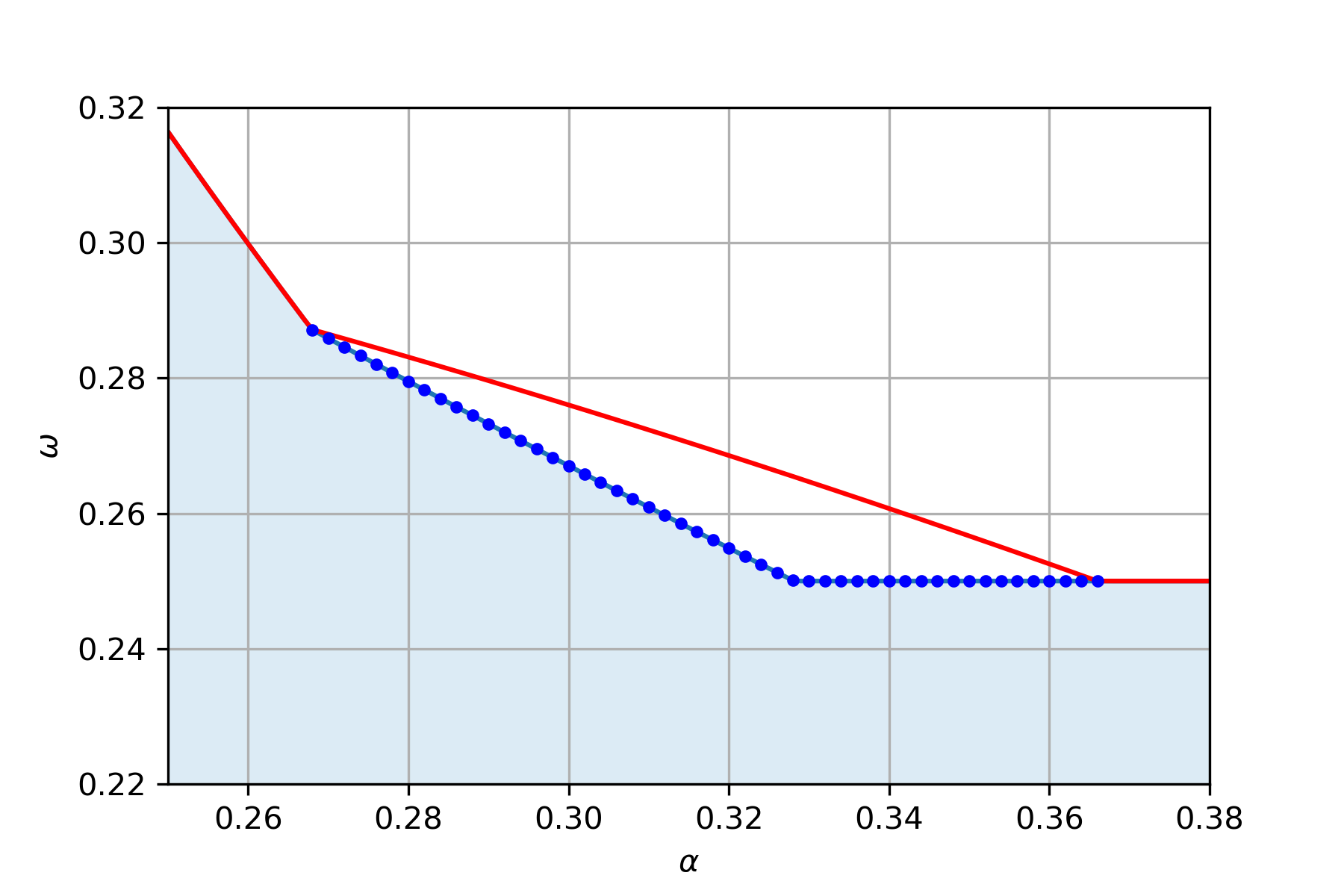}}
\caption{(a) Optimal classical (blue) and no-signalling (red) winning probabilities for the two-fold parallel repetition of the BSC game. The light blue area represents the values below the optimal classical winning probability. (b) Closeup of (a) with an additional numerical upper bound on the optimal quantum winning probability (blue dots) from the level ``$1+NM$'' of the NPA hierarchy for the values of $\alpha$ where the classical and no-signalling values differ. The numerical quantum upper bound is in excellent agreement with the classical value, suggesting its optimality (see \cref{conj:noadvantage}).}
\label{fig:twoCopies}
\end{figure}

In \cref{sec:channel_games} we introduce the notion of channel games, which are an extension of the LSSD problem in \cref{sec:BSC_game}. We  then define classical strategies based on codes and no-signalling strategies based on list-decoding schemes. In \cref{the:classical_exponennt} we provide an expression for the limit of the exponent of the classical winning probability, where we make use of strategies based on codes. Furthermore, we show that no-signalling strategies based on list-decoding schemes achieve the same limit as classical strategies. \edited{As a result, the optimal probability of winning for that class of LSSD games is asymptotically the same for all three types of resources available to the players. This allows one to solve the optimization problem for no-signalling strategies, for which there is an efficient algorithm, to find the asymptotic classical or quantum value which are otherwise computationally expensive to evaluate.}
%This result implies that no-signalling strategies based on list-decoding schemes are asymptotically optimal in the BSC game. 

\edited{As an extension, in \cref{sec:binLSSD} we analyze} three-party LSSD problems with \emph{binary} inputs and outputs. \Cref{lem:no-signal_multiple_parties} extends the two-party characterization from \cite{majenz2021local} for the classical winning probability of binary LSSD games to three parties. The main result of this appendix, \cref{thm:nogap}, shows that no-signalling resources cannot improve the winning probability of the players in this setting. 

\subsection{Open problems}

It would be interesting to examine the settings when there is a gap between the no-signalling and classical winning probabilities in the BSC game. Wherever there is a gap, it is interesting to look for a quantum strategy that also performs better than classical.

In the context of channel games, as introduced in \cref{sec:channel_games}, can we show, like for the BSC, that no-signalling strategies based on list-decoding schemes are asymptotically optimal? Are there more examples of channels for which there is a gap in winning probability between classical and no-signalling strategies in a finite number of parallel repetitions? Can the results be extended to classical-quantum channels where Alice and Bob receive a quantum state? For this last question, we would need to extend the idea of no-signalling to the case where the inputs and outputs can be quantum states.

\cref{sec:channel_games} also gives rise to a new area within information theory: simultaneous decoding. Within this setting, a sender tries to send a message to two receivers using identical channels, and the communication is successful if both receivers decode correctly. We can allow the receivers to share some quantum or no-signalling resources and examine whether this leads to better coding schemes. There are similar settings that have already been researched. In one such setting, the messages sent to the receivers are not necessarily the same, or two different channels are used (like in the book of El Gamal and Kim \cite[Part 2]{el2011network}). In another similar setting, we allow the sender and the receiver to share some entanglement (like in the book by Holevo \cite[Section 9]{holevo2019quantum}). There is even very recent research in a setting with two senders and one receiver that all share a no-signalling box (see the paper by Fawzi and Ferm\'e \cite{Ferme}). 

For the case of multi-player LSSD games with binary inputs and outputs as (see \cref{sec:binLSSD}), it is an open problem whether this result holds for any number of players. However, extending our numerical analysis to a larger number of players requires enumerating over all extrema of the corresponding no-signalling polytope. This polytope quickly grows in the number of vertices, making the analysis infeasible at the moment.

\section{Preliminaries}\label{sec:prelim}

For $n \in \N$, we denote the set $\{0, \dots, n-1\}$ by $[n]$ and the set of all permutations of $[n]$ by $S_n$. We denote by $\delta$ the indicator function, which is $1$ if its argument is true and $0$ otherwise. Throughout, we use binary logarithms and denote them by $\log$ rather than $\log_2$. We denote the bitwise XOR operator on bitstrings by $\oplus$ and the all-zero and all-one bitstrings of length $n$ by $0^n$ and $1^n$, respectively.
Let $X$ be a random variable over a finite set $\SX$. We denote its probability distribution by $P_\s{X}$ where $\s{X}$ is used to label the register that stores the random variable $X$. For any $n \geq 1$, we denote by $P_\s{X}^{\times n} = (P_{\s{X}})^{\times n}$ the product distribution of $n$ copies of $X$ on $\SX^n := \SX\times\dotsb\times\SX$ defined by
\[
    P_{\s{X}}^{\times n}(x^n) := \prod_{i=1}^n P_\s{X}(x_i),
\]
where $x^n = x_1 \dots x_n$ is an element of $\SX^n$. We sometimes omit writing the subscript in $P_\s{X}$, when it is obvious over which set $P$ is a distribution. For $\SA \subset \SX$, we denote by $P_\s{X}(\SA)$ the probability of random variable $X$ taking on a value in $\SA$:
\[
    P_\s{X}(\SA) = \sum_{x \in \SA} P_\s{X}(x).
\]
Lastly, for an arbitrary function $f:\SX \to \SY$, we define $f^{-1}(y) := 
\set{x\in\SX: f(x) = y}$.

\subsection{Quantum information}\label{sec:quantum}
A \emph{quantum state} on $\C^d$ is a $d \times d$ positive semi-definite matrix of unit trace, i.e., $\rho \in \C^{d \times d}$ such that $\rho \succeq 0$ and $\tr \rho = 1$.
We denote the set of all quantum states on $\C^d$ by $\D{\C^d}$.
Operations on quantum states are described by \emph{unitary} matrices, i.e., $U \in \C^{d \times d}$ such that $U\ct U = \Id$ where $\Id$ is the identity matrix.
We denote the set of all unitaries on $\C^d$ by $\U{\C^d}$.

An $n$-outcome \emph{measurement} or POVM on $\C^d$ is a collection of $n$ positive semi-definite $d \times d$ matrices that sum to identity.
We will denote a measurement by $M = \set{M_1, \dotsc, M_n}$ where $M_i \succeq 0$ and $\sum_{i=1}^n M_i = \Id$.
We denote the set of all $n$-outcome measurements on $\C^d$ by $\POVM{\C^d}$ (since the outcome set is always clear from the context, we do not specify it).
If $M_i^2 = M_i$ for all $i = 1, \dotsc, n$, we call the measurement \emph{projective}.
We denote the set of all $n$-outcome projective measurements on $\C^d$ by $\PM{\C^d}$.

\subsection{Linear programming}\label{sec:lin_prog}

Linear programming is a technique for optimizing a linear function over a convex polytope. A polytope is a generalization of a polygon to any number of dimensions. There are two ways of describing a convex polytope: by giving its extreme points (and rays), called the vertex representation, or by linear constraints, called the half-space representation.

The half-space representation of a convex polytope is a collection of (closed) half-spaces, such that their intersection is the convex polytope. A half-space can be described by a linear inequality
\begin{equation}\label{for:half-space}
    a_1 x_1 + \dots + a_n x_n \leq c.
\end{equation}
Using this description, the convex polytope can be represented as a system of linear inequalities, which can be written as a matrix inequality
\[
    Ax \leq d.
\]
Here, $A$ is the matrix containing all coefficients $a_i$ and $d$ the vector containing all constants~$c$, for all inequalities~\eqref{for:half-space} representing the polytope.
Note that we can also include linear equalities, as they can be described by two opposite inequalities.

Given a vertex representation, the corresponding convex polytope is the convex hull of the extreme points. The convex hull of a set of points is the smallest convex set that contains all the points, or simply the set of all convex combinations of the points (i.e., all weighted averages). This representation is especially interesting, since a linear function always has a global maximum in (at least) one of the extreme points of a convex polytope. We make use of this fact in \cref{sec:no_gap_three_player}.

\section{Local simultaneous state discrimination (LSSD)}\label{sec:LSSD}

In this section, we define the local simultaneous state discrimination (LSSD) task, originally introduced in \cite{majenz2021local}. In particular, we discuss strategies with classical, quantum and no-signalling resources for LSSD, and show that the optimal classical success probability can be attained by a symmetric strategy if certain conditions are fulfilled. Here we only consider the case of two players, Alice and Bob, but all definitions can easily be generalized to any number of players.

An LSSD game played by two players and a referee is defined by a classical-quantum-quantum (cqq) state~$\rho_\s{XAB}$, where the referee's register $\s{X}$ is classical while the Alice and Bob's registers $\s{A}$ and $\s{B}$ can generally be quantum. We denote the underlying spaces of $\s{X}$, $\s{A}$ and $\s{B}$ by $\HX = \C^{\SX}$, $\HA = \C^{\SA}$ and $\HB = \C^{\SB}$, respectively, where $\SX$, $\SA$ and $\SB$ are some finite sets.
We can always write the state $\rho_\s{XAB}$ as
\[
    \rho_\s{XAB} = \sum_{x \in \SX} P_\s{X}(x) \, \proj{x}_\s{X} \x \rho_\s{AB}^x,
\]
where $P_\s{X}$ is a probability distribution over $\SX$ and each $\rho^x_\s{AB}$ is a bipartite quantum state on $\HA \x \HB$.
% The referee gives the register $\s{A}$ to Alice and $\s{B}$ to Bob, while keeping the register $\s{X}$. 
The state $\rho_\s{XAB}$ is known to Alice and Bob, and they try to guess the referee's value $x$ based on their reduced states $\rho_\s{A}$ and $\rho_\s{B}$. We denote their guesses by $x_A$ and $x_B$. In general, Alice and Bob may share some additional resources before the game, but they are not allowed to communicate with each other during the game. They win the game if both guesses are correct: $x_A = x_B = x$.

In most of this paper, we are going to consider the case where $\rho_\s{XAB}$ is entirely classical. Meaning that there exists an orthonormal basis $\set{\ket{a} : a \in \SA}$ of $\HA$ and $\set{\ket{b} : b \in \SB}$ of $\HB$ that are independent of $x \in \SX$, and probability distributions $P_\s{AB}^x$ over $\SA \times \SB$ such that
\[
    \rho_\s{AB}^x = \sum_{\substack{a\in\SA\\b\in\SB}} P_\s{AB}^x(a,b) \, \proj{a}_\s{A} \otimes \proj{b}_\s{B}.
\]

In this case, it is useful to rephrase the problem. Instead of describing the game by a cqq state, we can describe it by a probability distribution $P_\s{XAB}$ on $\SX \times \SA \times \SB$. The referee picks elements $x \in \SX, a\in\SA$ and $b \in \SB$ according to this distribution and gives $a$ and $b$ to Alice and Bob, respectively. Alice and Bob know the distribution $P_\s{XAB}$ and both try to guess the value $x$. Again, they may share some resources, but are not allowed to communicate during the game, and they win if they both guess $x$ correctly. A schematic representation of LSSD is shown in~\cref{fig:LSSD_schematic}.

\begin{figure}
    \centering
    \includegraphics[scale=0.4]{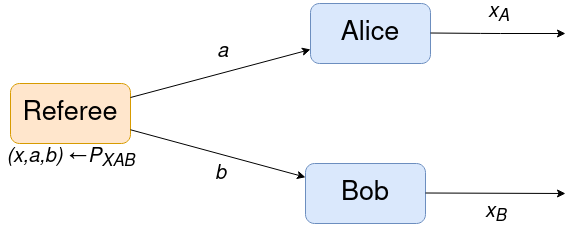}
    \caption{A schematic of the LSSD game. On inputs $a$ and $b$, Alice and Bob make guesses $x_A$ and $x_B$ respectively, and win if $x=x_A=x_B$.}
    \label{fig:LSSD_schematic}
\end{figure}

We now describe different types of strategies based on three different possible shared resources: classical, quantum and no-signalling.
While these additional resources can be of different types, the strategies themselves are in general quantum since the LSSD game is based on a quantum state.

\subsection{Classical resources}\label{sec:Classical}

While strategies for LSSD may in general take advantage of shared randomness, this does not help in increasing the winning probability. Indeed, after a random value is generated, we are left with a deterministic strategy that depends on this value. Thus instead of the original randomized strategy, the players can just use one of the deterministic strategies that achieves the highest winning probability. Hence in the following, we assume that the players do not use shared randomness.

In the quantum case of the LSSD game (meaning that the game is described by a cqq state $\rho_\s{XAB}$), a strategy is completely defined by two measurements $M = \set{M_x : x \in \SX}$ and $N = \set{N_x : x \in \SX}$ on $\HA$ and $\HB$, respectively. Alice and Bob perform these measurements on their subsystems to produce their guesses for $x$. Given the measurements $M$ and $N$, their winning probability is
\[
    \sum_{x\in\SX} P_\s{X}(x) \tr[\rho_\s{AB}^x(M_x \otimes N_x)],
\]
and the optimal winning probability is denoted by
\[
    \omegac (\s{X}|\s{A};\s{B})_\rho := \sup_{\substack{M \in \POVM{\HA}\\N\in \POVM{\HB}}} \sum_{x\in\SX} P_\s{X}(x) \tr[\rho_\s{AB}^x(M_x \otimes N_x)],
\]
where $\POVM{\HA}$ and $\POVM{\HB}$ denote the sets of all measurements on $\HA$ and $\HB$, respectively.

In case $\rho_\s{XAB}$ is purely classical and described by a probability distribution $P_\s{XAB}$, the strategy of Alice and Bob is given by two conditional probability distributions $Q_{\s{X}_A|\s{A}}$ and $Q_{\s{X}_B|\s{B}}$ describing their local behaviour. The winning probability is then given by 
\[
    \sum_{\substack{x\in\SX\\a\in\SA, b\in\SB}} P_\s{XAB}(x, a, b) Q_{\s{X}_A|\s{A}}(x|a) Q_{\s{X}_B|\s{B}}(x|b).
\]
The optimal winning probability can now be obtained by maximizing over all conditional probabilities. However, we can restrict this optimization to maximizing over all deterministic strategies, i.e., strategies that can be described by two functions $f \colon \SA \to \SX$ and $g \colon \SB \to \SX$. Similarly to shared randomness, Alice and Bob can condition any local randomness on the realization that maximizes their probability of winning. Now, the optimal winning probability is given by 
\[
    \omegac(\s{X}|\s{A};\s{B})_P = \max_{f,g} \sum_{\substack{x\in\SX\\a\in\SA, b\in\SB}}P_\s{XAB}(x, a, b) \delta[f(a) = g(b) = x].
\]

We say that a strategy is \emph{symmetric} if Alice and Bob perform the same local strategy, i.e., if $f = g$. In the following theorem, we show that symmetric strategies attain optimal classical values for classical LSSD games (see \cref{sec: proot theorem symmetry} for proof).

\begin{theorem}\label{the:symmetry}
Let $P_\s{XAB}$ be a distribution over $\SX \times \SA \times \SB$, with $\SA = \SB$, satisfying the following:
\begin{enumerate}[(i)]
    \item The marginal distribution $P_\s{X}$ over $\SX$ is uniform.
    \item $P_\s{AB|X} = P_\s{A|X}P_\s{B|X}$.
    \item $P_\s{A|X} = P_\s{B|X}$.
\end{enumerate}
Then the classical LSSD game defined by $P_\s{XAB}$ has an optimal deterministic strategy that is symmetric.
\end{theorem}

\subsection{Quantum resources}\label{sec:quantum resources}

In this case, Alice and Bob can share an entangled state prior to receiving their inputs. Let $\Ha = \Hb = \C^d$ be two complex Euclidean spaces of dimension $d$.
Alice and Bob first jointly prepare a quantum state $\sigma_\ab$ on $\Ha \x \Hb$, after which Alice and Bob keep systems $\a$ and $\b$, respectively.
After receiving their inputs, Alice and Bob determine their output by measuring the registers $\Aa$ and $\Bb$ with local measurements $M$ and $N$, respectively (this is the most general strategy because no communication is allowed).

When the local dimensions of the shared entangled state $\sigma_\ab$ are limited to $d$ for both parties, the optimal probability of winning is
\begin{align}
  \pg[d]{q}{\X}{\A;\B}{\rho} \eqdef
  \sup_{\substack{\sigma_\ab \in \D{\C^\da \x \C^\db}}}
  \sup_{\substack{M \in \POVM{\HA \x \C^\da} \\ N \in \POVM{\HB \x \C^\db}}}
  \sum_{x \in \SX} P_\X(x)
  \tr \sof[\big]{\of{\rho_{\AB}^x \x \sigma_\ab} \of{M_x \x N_x}}.
  \label{eq:pd}
\end{align}
When the dimensions of $\a$ and $\b$ are not limited, the optimal winning probability is
\begin{align}
	\pg{q}{\X}{\A;\B}{\rho} \eqdef
  \sup_{d \geq 1} \pg[d]{q}{\X}{\A;\B}{\rho}.
  \label{eq:pq def}
\end{align}
When $\rho_{\XAB}$ is classical and described by a probability distribution $P_{\XAB}$, we can simplify \cref{eq:pd} as follows:
\begin{align}
  \pg[d]{q}{\X}{\A;\B}{P}
 &= \sup_{\sigma_\ab \in \D{\C^\da \x \C^\db}}
    \sup_{\substack{M: \SA \to \POVM{\C^{\da}} \\ N: \SB \to \POVM{\C^{\db}}}}
    \sum_{\substack{x\in \SX \\ a \in \SA , b \in \SB}} P_{\XAB}(x, a, b)
    \tr \sof[\big]{\sigma_\ab \of[\big]{M_x(a) \otimes N_x(b)}}
    \label{eq:quantum value} \\
 &= \sup_{\substack{M: \SA \to \POVM{\C^{\da}} \\ N: \SB \to \POVM{\C^{\db}}}}
    \norm[\bigg]{\sum_{\substack{x\in \SX \\ a \in \SA , b \in \SB}} P_{\XAB}(x,a,b) M_x(a) \x N_x(b)},
    \label{eq:pq for distribution}
\end{align}
where $M$ and $N$
are collections of measurements,
i.e., for every input $a \in \SA$ and $b \in \SB$,
we have that
$M(a) = \set{M_x(a) : x \in \SX}$ and
$N(b) = \set{N_x(b) : x \in \SX}$
are measurements on $\C^d$ with outcomes in $\SX$.
%We show in \cref{cor:projective} that the optimization in $\pg{q}{\X}{\A;\B}{P}$ can be restricted to projective measurements.

\subsection{No-signalling resources}

We define strategies with no-signaling resources only when $\rho_{\XAB}$ is classical and described by a probability distribution $P_{\XAB}$. Given classical inputs $a \in \SA$ and $b \in \SB$ for Alice and Bob, respectively, they output their estimates $x_A$ and $x_B$ of $x \in \SX$ according to a conditional probability distribution $Q_{\X_A\X_B|\AB}$ on $\SX \times \SX \times \SA \times \SB$ satisfying
\begin{align}
  \forall x_B, a, a', b: \quad
    \sum_{x_A \in \SX} Q_{\X_A\X_B|\AB}(x_A, x_B | a, b)
 &= \sum_{x_A \in \SX} Q_{\X_A\X_B|\AB}(x_A, x_B | a', b), \label{eq:ns1} \\
  \forall x_A, a, b, b': \quad
    \sum_{x_B \in \SX} Q_{\X_A\X_B|\AB}(x_A, x_B | a, b)
 &= \sum_{x_B \in \SX} Q_{\X_A\X_B|\AB}(x_A, x_B | a, b'). \label{eq:ns2}
\end{align}
An optimal no-signaling strategy succeeds with probability
\begin{align}
  \pg{ns}{\X}{\A;\B}{P} \eqdef
  \sup_{Q_{\X_A\X_B|\AB}}
  \sum_{\substack{x\in \SX \\ a \in \SA , b \in \SB}}
  P_{\XAB}(x,a,b)
  Q_{\X_A\X_B|\AB}(x,x|a,b).
  \label{eq:pns def}
\end{align}
% where the conditional probability distribution that defines the strategy must satisfy the following \emph{no-signalling constraints}:
% \begin{align}
%     \sum_{x_A \in \SX} Q_{\X_A\X_B|\AB}(x_A,x_B|a_1,b)
%   = \sum_{x_A \in \SX} Q_{\X_A\X_B|\AB}(x_A,x_B|a_2,b),
%     \quad
%     \forall x_B \in \SX, a_1,a_2 \in \SA, b \in \SB.    
% \end{align}

The set of classical correlations is a subset of the set of quantum correlations, and the latter is a subset of the set of no-signalling correlations, see \cite{Bell_non_locality_report} for more details. Therefore, we have that
\begin{equation}\label{eq:pgrelations}
    \pg{c}{\X}{\A;\B}{P} \leq \pg{q}{\X}{\A;\B}{P} \leq \pg{ns}{\X}{\A;\B}{P}.
\end{equation}

Notice that the winning probability for a given no-signalling strategy is a linear function in the values $Q_{\s{X}_A\s{X}_B|\s{AB}}(x_A,x_B|a,b)$. This, together with the fact that the set of no-signalling correlations forms a convex polytope, see e.g.~\cite{Bell_non_locality_report}, implies that we can use linear programming to find the optimal no-signalling winning probability of an LSSD game. It also implies that there is always an optimal strategy at one of the extreme points of the no-signalling polytope.

This last fact is what Majenz et al.~used to prove that there exists no probability distribution $P_\s{XAB}$ with binary $x,a$ and $b$, such that the corresponding LSSD game can be won with higher probability using no-signalling strategies \cite[Proposition~3.3]{majenz2021local}. They showed that none of the no-signalling correlations at the extreme points of the no-signalling polytope could ever perform better than the simple classical strategy of outputting the most likely value for $x$. We do something similar in \cref{sec:binLSSD} for the tripartite case. However, it turns out that this argument is not enough in the tripartite case, and we take a numerical approach to finish the argument.

\section{The binary-symmetric-channel game}\label{sec:BSC_game}

A binary symmetric channel (BSC) with error $\alpha\in[0,1/2]$ is a channel with a single bit of input that transmits the bit without error with probability $1-\alpha$ and flips it with probability $\alpha$, see \cref{fig:BSC}. In this section, we study a particular LSSD problem: the binary-symmetric-channel game, originally introduced in \cite[Example 1]{majenz2021local}, where a referee sends a bit to Alice and Bob over two identical and independent binary symmetric channels, both with error probability $\alpha$, see \cref{def:BSCgame} for a formal definition. In \cite{majenz2021local}, an explicit optimal classical strategy for this game is shown and its corresponding optimal winning probability for every $\alpha$ is obtained. Moreover, the authors show that the winning probability cannot be improved by any quantum nor no-signalling strategy. In addition, they show that if two copies of the game are played in parallel for $\alpha=1-\frac{1}{\sqrt{2}}$, there is an explicit optimal classical strategy that performs better than repeating the optimal classical strategy for a single copy of the game twice and, as a consequence, quantum and no-signalling optimal strategies must perform better than repeating the respective optimal strategies for a single copy of the game.

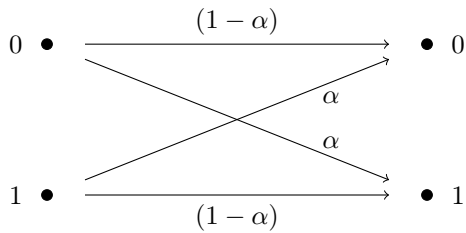
\begin{figure}[ht]
\centering
\begin{tikzpicture}
\filldraw[black] (0,0) circle (2pt) node[anchor=east, xshift=-0.2cm]{1};
\filldraw[black] (0,2) circle (2pt) node[anchor=east ,xshift=-0.2cm]{0};
\filldraw[black] (5,0) circle (2pt) node[anchor=west, xshift=0.2cm]{1};
\filldraw[black] (5,2) circle (2pt) node[anchor=west, xshift=0.2cm]{0};
\draw[->] (0.5, 0) -- (4.5, 0) node[anchor=north, xshift=-2.0cm]{$(1-\alpha)$};
\draw[->] (0.5, 2) -- (4.5,2) node[anchor=south, xshift=-2.0cm]{$(1-\alpha)$};
\draw[->] (0.5, 0.2) -- (4.5, 1.8) node[anchor=west, yshift=-0.5cm, xshift=-1.0cm]{$\alpha$};
\draw[->] (0.5, 1.8) -- (4.5, 0.2) node[anchor=west, yshift=0.5cm, xshift=-1.0cm]{$\alpha$};
\end{tikzpicture}
\caption{Schematic representation of a binary symmetric channel with error probability $\alpha$.}
\label{fig:BSC}
\end{figure}

In \cref{sec:parallelrepetitionBSCgame}, we study the parallel repetition of the BSC game and, for the case of two copies, we provide the optimal classical, quantum and no-signalling values, showing that for most $\alpha$ the three values coincide (and in most of the cases the optimal values are obtained just by repeating the optimal strategy for a single copy of the BSC game). Nevertheless, for certain values of $\alpha$, the classical and quantum values coincide but there is a no-signalling advantage.

In \cref{sec:threecopiesBSC}, we provide the optimal no-signalling winning probabilities for the three-fold parallel repetition of the BSC game. We study the `good' classical and no-signalling strategies for arbitrary number $n$ of parallel rounds of the BSC game in \cref{sec: arbitrary copies BSC}.

\begin{definition}[Example~1 in \cite{majenz2021local}]\label{def:BSCgame}
Let $X,Y$ and $Z$ be independent binary random variables such that $X$ is uniformly random, i.e., $\Pr[X=1] = 1/2$, and $\Pr[Y=1] = \Pr[Z=1] = \alpha$ for $\alpha \in [0,1/2]$. Let $A:=X\oplus Y$ and $B:=X\oplus Z$, and denote the joint probability mass function of $(X,A,B)$ by $P^{\alpha}_{\s{XAB}}$. The \emph{binary-symmetric-channel} (BSC) game is defined as the task of simultaneously guessing $X$ from $A$ and $B$.
\end{definition}

\begin{prop}[Example~1 in \cite{majenz2021local}]
For every $\alpha\in[0,1/2]$, the optimal classical, quantum and no-signalling winning probabilities for the BSC game $P^\alpha$ are equal and given by
\begin{equation}\label{eq w_c=w_q=w_ns BSC}
    \omegac(\s{X|A;B})_{P^{\alpha}} =
    \omegaq(\s{X|A;B})_{P^{\alpha}} =
    \omegans(\s{X|A;B})_{P^{\alpha}} =
    \begin{cases}
        (1 - \alpha)^2 & \textrm{ if } \alpha\in[0,1 - \frac{1}{\sqrt{2}}],
        %0 \leq \alpha \leq 1 - \frac{1}{\sqrt{2}},
\\
        \frac{1}{2} & \textrm{ if } \alpha\in(1 - \frac{1}{\sqrt{2}},\frac{1}{2}].
        % 1 - \frac{1}{\sqrt{2}}< \alpha \leq \frac{1}{2}.
    \end{cases}
\end{equation}
\end{prop}

The optimal winning probability for $\alpha \in [0,1-1/\sqrt{2}]$ is achieved by the strategy where Alice and Bob output the input they received. The intuition behind this strategy is that for `small' $\alpha$, the bits they receive most likely have not been flipped. Notice that if Alice and Bob were playing this game without having to coordinate their answers, such a strategy would be optimal for all $\alpha$. In fact, the optimal strategy for `high'-noise BSC channels, $\alpha \in (1-1/\sqrt{2}, 1/2]$, is achieved by both parties outputting some previously agreed bit. 

\subsection{Two-fold parallel repetition of the binary-symmetric-channel game} \label{sec:parallelrepetitionBSCgame}

Let $(X',A',B')$ be an independent copy of $(X,A,B)$, as described in \cref{def:BSCgame}. The two-fold parallel repetition of the BSC game consists of simultaneously guessing $(X,X')$ from $(A,A')$ and $(B,B')$.
This game is described by the probability distribution $P^{\alpha}_{\s{XAB}} \otimes P^{\alpha}_{\s{X'A'B'}}$.
According to \cite{majenz2021local}, the optimal classical winning probability for the two-fold parallel repetition of the BSC game for $\alpha=1-\frac{1}{\sqrt{2}}$ is
\begin{equation}
    \frac{1}{4}(1-\alpha^2)^2+\frac{1}{4}(1-\alpha)^4.
\end{equation}
Hence, for $\alpha=1-\frac{1}{\sqrt{2}}$, $\omegac(\s{X X'| AA';BB'})_{P^{\alpha}\otimes P^{\alpha}}>\omegac(\s{X | A;B})_{P^{\alpha}}^2$ and, from \eqref{eq:pgrelations} and \eqref{eq w_c=w_q=w_ns BSC}, we also have 
\begin{equation}
    \omegaq(\s{X X'| AA';BB'})_{P^{\alpha}\otimes P^{\alpha}}>\omegaq(\s{X | A;B})_{P^{\alpha}}^2,
\end{equation}
\begin{equation}
    \omegans(\s{X X'| AA';BB'})_{P^{\alpha}\otimes P^{\alpha}}>\omegans(\s{X | A;B})_{P^{\alpha}}^2.
\end{equation}

\edited{Here we study the full range of $\alpha$ (namely, $\alpha \in [0,1/2]$). In the following Theorem, we provide the optimal classical and no-signalling winning probabilities for the two-fold parallel repetition of the BSC game, graphically represented in \cref{fig:twoCopies}. The Theorem shows that for most values of $\alpha$, the classical and no-signalling optimal success probabilities coincide (and therefore so does the quantum value).}
%Since the BSC game fulfills the conditions of Theorem \ref{the:symmetry}, a symmetric strategy will provide the optimal classical value. By brute force over all symmetric classical strategies and by solving a linear problem (see \cite{Code} for the computer code), we obtain the optimal classical and no-signalling winning probabilities of the BSC game for every value of $\alpha$, see Result \ref{result w_c and w_ns repetition}. The result shows that for most of the values of $\alpha$, the classical and no-signalling optimal success winning probabilities for a parallel repetition of the BSC game coincide and therefore so does the quantum. Result \ref{result w_c and w_ns repetition} is represented in Fig.~\ref{fig:twoCopies}. 

 \begin{theorem}\label{result w_c and w_ns repetition} Let $(X',A',B')$ be an independent copy of $(X,A,B)$. Let $\alpha_0 < 1$ be the real solution of $(1-\alpha^2)^2+(1-\alpha)^4=1$, i.e.~$\alpha_0\simeq0.32814$, and let $I_1=[0,2-\sqrt{3}]$, $I_2=(2-\sqrt{3},\alpha_0]$, $I_3=(\alpha_0,\frac{\sqrt{3}-1}{2}]$ and $I_4=(\frac{\sqrt{3}-1}{2},\frac{1}{2}]$. Then, for the two-fold parallel repetition of the BSC game, we have
\begin{equation}\label{eq wc repetition BSC}
    \omegac(\s{XX'|AA';BB'})_{P^{\alpha}\otimes P^{\alpha}} =
    \begin{cases}
        (1 - \alpha)^4 & \textrm{ if }\alpha\in I_1, \\
        \frac{1}{4}(1-\alpha^2)^2+\frac{1}{4}(1-\alpha)^4 & \textrm{ if } \alpha\in I_2, \\
        \frac{1}{4} & \textrm{ if } \alpha\in I_3\cup I_4,
    \end{cases}
\end{equation}
and 
\begin{equation}\label{eq wns repetition BSC}
    \omegans(\s{XX'|AA';BB'})_{P^{\alpha}\otimes P^{\alpha}} =
    \begin{cases}
        (1 - \alpha)^4 & \textrm{ if } \alpha\in I_1, \\
        \frac{(1-\alpha^2)^2}{3}& \textrm{ if } \alpha\in I_2 \cup I_3, \\
        \frac{1}{4} &\textrm{ if }  \alpha\in I_4.
    \end{cases}
\end{equation}
\end{theorem}

\begin{proof}
\edited{Since the BSC game fulfills the conditions of \cref{the:symmetry}, a symmetric strategy will provide the optimal classical value. We determine $\omegac$ by considering all deterministic classical strategies.
For each strategy, we compute the winning probability as a function of $\alpha$. Then we obtain the analytical value~\eqref{eq wc repetition BSC} by taking the maximum and applying the \WolframRef{PiecewiseExpand} command. For more details on this derivation, see the \textit{Mathematica} file \file{BSC classical strategy n=2.nb} in \cite{code}.}

\newcommand{\eq}{\mathrm{eq}}

\edited{The optimal no-signalling value can be found via a linear program, i.e., a maximization of a linear function subject to linear constraints. In \textit{Mathematica}, the standard form to represent a linear program that optimizes over $x\in\R^n$ is}
\edited{\begin{equation}\label{eq:primal}
    \begin{split}
        \text{\emph{Primal problem:}}\hspace{2cm}
        \text{minimize: }
            & \langle c,x \rangle = \sum_{i=1}^nc_ix_i \\
        \text{subject to: }
            & A x + b \geq 0, \\
            & A_{\eq} x + b_{\eq} = 0,
    \end{split}
\end{equation}}
\edited{where $x,c\in\R^n$, $A\in\R^{m\times n}$, $b\in\R^m$, $A_{\eq}\in\R^{k\times n}$, $b_{\eq}\in\R^k$ (see \WolframRef{LinearOptimization} for more details).
Its dual, which optimizes over $\lambda\in\R^m$ and $\nu\in\R^k$, is given by
\begin{equation}\label{eq:dual}
    \begin{split}
        \text{\emph{Dual problem:}}\hspace{1cm}
        \text{maximize: }
            & -\of[\big]{\langle b,\lambda\rangle+\langle b_{\eq},\nu\rangle} = -\sum_{i=1}^m b_i\lambda_i-\sum_{i=1}^k b_{\eq,i}\nu_i \\
        \text{subject to: }
            & A\tp\lambda + A_{\eq}\tp \nu - c = 0, \\
            & \lambda \geq 0.
    \end{split}
\end{equation}
}

\edited{A common technique in linear programming is to use one of the two problems to obtain a bound on the other.
In the above formulation, any feasible solution to the dual problem~\eqref{eq:dual} provides a lower bound on the optimal solution of the primal problem~\eqref{eq:primal}.
The optimal value of both problems can be determined by finding feasible primal and dual solutions that have the same value. Then, as a consequence of strong duality, both solutions must be optimal.}

\edited{
Since the original linear program for computing $\omegans$ for the BSC game is quite large, see \cref{eq:ns1,eq:ns2,eq:pns def}, we first simplify it by reducing the number of parameters. We do this by imposing the following symmetries on Alice's and Bob's no-signalling strategy $Q$:\footnote{Here we consider only two parallel repetitions of the BSC game. But the same symmetry reductions can be performed for any number of repetitions (see \cref{result 3 copies NS}).}
\begin{enumerate}
    \item By \cref{lem:permutations} below, there is an optimal no-signalling strategy that is invariant under any permutation of the instances of the game, i.e., $Q\of[\big]{\sigma(x),\sigma(y)|\sigma(a),\sigma(b)} = Q(x,y|a,b)$, for any permutation $\sigma$ of positions within a string.
    \item Since the BSC game is symmetric under exchanging Alice and Bob, we can also exchange Alice's and Bob's strategies, i.e., $Q(y,x|b,a) = Q(x,y|a,b)$.
    \item Since the BSC game is symmetric under negating any subset of input and output bits, we can do the same to Alice's and Bob's strategy, i.e.,
    $Q(x \oplus s, y \oplus s | a \oplus s, b \oplus s) = Q(x,y|a,b)$ for any bit string $s$.
\end{enumerate}
}

\edited{
After performing the above symmetry reductions, we need to find feasible primal and dual solutions of equal value.
These solutions should be $\alpha$-dependent, i.e., work not just for a single value of $\alpha$ but for whole intervals of $\alpha$.
We managed to find such solutions with the help of \textit{Mathematica}, and we have provided them in the format of \cref{eq:primal,eq:dual} in the notebook \file{BSC no-signalling strategy n=2.nb} \cite{code}.
The primal and dual objective values of these solutions match and agree with \cref{eq wns repetition BSC} in each of the intervals $I_1,\dotsc,I_4$ (occasionally we could not obtain a single $\alpha$-dependent solution for a whole interval, in which case we broke it into smaller subintervals).
}

\edited{
Finding these exact $\alpha$-dependent solutions required some numerical tricks.
Indeed, while it is easy to solve the linear program for any particular value of $\alpha$, obtaining continuous $\alpha$-dependent solutions is nontrivial -- it requires interpolating from a small number of solutions, or often even a single solution.
We used a combination of the following numerical tricks to cover all cases in \cref{eq wns repetition BSC} (often obtaining the same solution with different methods):
\begin{itemize}
    \item \emph{Rational multiples of $\pi$}:
    We chose a rational number $r$ so that $\alpha = r \pi$ lies in a given interval~$I_i$.
    Using \WolframRef{LinearOptimization} we then find a symbolic solution that is polynomial in $\pi$.\footnote{This works since on one hand \textit{Mathematica} treats $\pi$ symbolically, while on the other it can compare $\pi$ to any other number by calculating its numerical value to arbitrary accuracy. It is also important that \textit{Mathematica} can manipulate rational numbers symbolically and that $\pi$ is irrational.}
    Substituting back $\pi = \alpha/r$ gives us an exact polynomial $\alpha$-dependent solution.
    This is quite remarkable since we have effectively interpolated a polynomial function from a single irrational point.
    This strategy unfortunately did not work for $3$ repetitions of the game since the linear program was too large.
    \item \emph{Rational solutions}:
    We choose a sequence of equally spaced rational values of $\alpha$ and find exact rational solutions for these values by using \WolframRef{LinearOptimization}.
    We then interpolate between them by using \WolframRef{FindSequenceFunction}.
    This method generally requires some fiddling with the chosen sequence since nearby values of $\alpha$ can lead to completely different and unrelated solutions.
    \item \emph{Algebraic solutions}:
    We choose an algebraic $\alpha$ from the given interval $I_i$ and find a numerical solution for this $\alpha$ to extremely high accuracy ($300$ digits).
    Then we use \WolframRef{RootApproximant} to turn this numerical solution into exact algebraic numbers.
    Reconstructing the minimal polynomial for each of these numbers gives us an interpolated $\alpha$-dependent solution that is polynomial.
    This trick effectively interpolates from a single algebraic point.
\end{itemize}
Checking the primal and dual constraints of the resulting interpolated solution gives us constraints on $\alpha$ that capture the interval in which this solution holds.}

\edited{
It is important to note that, irrespective of how dirty the above numerical methods are, once an exact $\alpha$-dependent solution is found, it can be easily verified that it satisfies all constraints and gives equal primal and dual values, hence implying optimality.
For more details, see \file{BSC no-signalling strategy n=2.nb} in \cite{code}.}
\end{proof}

Notice that, unlike a single copy of the BSC game, the optimal winning probabilities have different behaviors split into three different intervals. We see that 
\begin{equation}
    \omegac(\s{XX'|AA';BB'})_{P^{\alpha}\otimes P^{\alpha}} = \omegans(\s{XX'|AA';BB'})_{P^{\alpha}\otimes P^{\alpha}}=\omegac(\s{X|A;B})_{P^{\alpha}}^2\hspace{0.5cm} \forall \alpha \in I_1 \cup I_4,
\end{equation}
and therefore, due to \eqref{eq:pgrelations}, the quantum value is the same value as the classical. Analogously to the single copy of the BSC game, for `small' $\alpha$, $\alpha\in I_1$, an optimal classical and no-signalling strategy is given by Alice and Bob outputting their input. The intuition behind it is that, due to `low' noise, every bit has low probability of being flipped, $(1-\alpha)$, and thus the winning probability using this strategy is $(1-\alpha)^4$. On the other hand, an optimal classical and no-signalling strategy for a `high' noisy channel,  $\alpha\in I_3\cup I_4$ and $\alpha\in I_4$, respectively, is that both Alice and Bob output some previously agreed bit string. This leads to the conclusion that the corresponding optimal winning probabilities for these values of $\alpha$ can be achieved by just repeating the optimal classical and no-signalling strategies mentioned above for a single copy of the BSC game. Nevertheless, this is not always the case, since
\begin{equation}\label{eq c<ns for I2 U I3}
    \omegac(\s{XX'|AA';BB'})_{P^{\alpha}\otimes P^{\alpha}} < \omegans(\s{XX'|AA';BB'})_{P^{\alpha}\otimes P^{\alpha}} \hspace{0.5cm} \forall\alpha\in I_2\cup I_3.
\end{equation}
An optimal classical strategy for $\alpha\in I_2$ is given by Alice and Bob both outputting $00$ if their input contains a $0$ and outputting $11$, otherwise, which gives an optimal winning probability of $\frac{1}{4}(1-\alpha^2)^2+\frac{1}{4}(1-\alpha)^4$, which was already given in \cite{majenz2021local} for $\alpha=1-\frac{1}{\sqrt{2}}$. An optimal no-signalling strategy for $\alpha\in I_2\cup I_3$ is given by 
\begin{align}\label{two_copy_ns_strat}
    Q_2(x,y|a,b) = \begin{cases}
    \frac{1}{3} & \text{if } (x=y \text{ or } x\oplus b = 11 = y\oplus a) \text{ and } (x \oplus a \not= 11 \not= y\oplus b),\\
    0 & \text{otherwise.}
    \end{cases}
\end{align}
This strategy, see \cref{n_copies_section}, has winning probability $(1-\alpha^2)^2/3$. More specifically, for $\alpha\in I_2$ and for $\alpha\in I_2\cup I_3$ there exist classical and no-signalling strategies, respectively, that perform better than repeating the optimal strategy, i.e.
\begin{equation}\begin{split}
    \omegac(\s{XX'|AA';BB'})_{P^{\alpha}\otimes P^{\alpha}}&>\omegac(\s{X|A;B})_{P^{\alpha}}^2 \textrm{ }\forall \alpha\in I_2,\\
    \omegans(\s{XX'|AA';BB'})_{P^{\alpha}\otimes P^{\alpha}}&>\omegans(\s{X|A;B})_{P^{\alpha}}^2 \textrm{ }\forall \alpha\in I_2\cup I_3.
    \end{split}
\end{equation}

We are left with characterizing the value $\omegaq(\s{XX'|AA';BB'})_{P^{\alpha}\otimes P^{\alpha}} $ for $\alpha\in I_2\cup I_3$. From \eqref{eq c<ns for I2 U I3}, the optimal quantum value for $\alpha\in I_2\cup I_3$ has to be in between the two values.
\edited{Based on strong numerical evidence (see \cref{fig:twoCopies}), in \cref{conj:noadvantage} below we conjecture that there is no quantum advantage with over the optimal classical strategy for any $\alpha$.}

Unlike the set of classical and the set of no-signaling correlations, the set of quantum correlations, $\mathcal{Q}$, has uncountably many extremal points, see e.g.~\cite{Bell_non_locality_report}, making the optimization problem a tough task. In \cite{NPA2008}, Navascués, Pironio and Acín (NPA) introduced an infinite hierarchy of conditions necessarily satisfied by any set of quantum correlations with the property that each of them can be tested using semidefinite programming (SDP) and thus they can be used to exclude non-quantum correlations, see \cref{section Appendix NPA}. The authors introduced a recursive way to construct subsets ${\mathcal{Q}_{\ell}\supset\mathcal{Q}_{\ell+1}\supset\mathcal{Q}}$ for all $\ell\in\N$, each of them can be tested using semidefinite programming and are such that $\cap_{\ell\in\N}\mathcal{Q}_{\ell}=\mathcal{Q}$, i.e.~they converge to the set of quantum correlations. 

By using an intermediate level between the first and the second levels of the NPA hierarchy, the so-called level ``$1+MN$'' (see \cref{section Appendix NPA} for a detailed explanation and \file{NPA\_hierarchy\_BSC\_Game.py} \cite{code} for the numerical code), we find that for $\alpha\in I_2$,  $\omegaq(\s{XX'|AA';BB'})_{P^{\alpha}\otimes P^{\alpha}}$ is upper bounded by $\omegac(\s{XX'|AA';BB'})_{P^{\alpha}\otimes P^{\alpha}}$, see \cref{fig:twoCopies}~(b). Therefore, this shows that the values coincide in the interval $I_2$. The reason to restrict ourselves to the level ``$1+MN$'' is that it requires less computational resources than computing the level $2$ and it already provides tight bounds. \edited{Based on the fact that the numerical upper bounds on the quantum value obtained by solving the semidefinite programs match the (analytical) lower bounds given by the classical values, we state the following conjecture.}

\begin{conjecture}\label{conj:noadvantage}
There is no quantum advantage over the best classical strategy for the two-fold parallel repetition of the BSC game for any value of $\alpha$.
\end{conjecture}

\subsection{Three-fold parallel repetition of the BSC game} \label{sec:threecopiesBSC}

Consider the three-fold parallel repetition of the BSC game. \edited{In the following Theorem, we provide the optimal classical and no-signalling winning probabilities, and we will see that for a vast range of values of $\alpha$ they coincide and therefore so does the quantum. }

%\edited{THE FOLLOWING NO LONGER APPLIES, WE HAVE IT:} In such a case, if we only consider symmetric strategies, there are $8^8 = 2^{24}$ possibilities, which makes the task of finding an optimal deterministic strategy inefficient. This does not only make the optimization task tough for a given $\alpha$ but it makes it even harder if we consider a set of possible values of $\alpha$.  We therefore take another approach: we first consider no-signalling strategies\edited{, I'd remove this:  which can be efficiently found using a linear program, } and then we find classical strategies that attain the optimal no-signalling values. In this way, we show that for most of the values of $\alpha$, there is no no-signalling advantage with respect to the best classical strategy and therefore there is also no quantum advantage. 

%In a similar way as for the two-fold parallel repetition of the BSC game, we obtain the optimal no-signalling winning probabilities solving a linear program, see \cite{code} for the code. The results are shown in Result \ref{result 3 copies NS}.

\begin{theorem}\label{result 3 copies NS} Let $(X',A',B')$ and $(X'',A'',B'')$ be two independent copies of $(X,A,B)$ and let $\alpha_1$ be the root of the polynomial $2(1-\alpha)^4(1+2\alpha)-1$ taking the value $\alpha_1\simeq0.358121$,  $\alpha_2=\frac{1}{8}(3-\sqrt{7}+\sqrt{2(32-11\sqrt{7})})$ and $\alpha_3=2^{-\frac{2}{3}}(4-\sqrt{14})^{\frac{1}{3}}$. Then, for three copies of the BSC game, 
\begin{equation}\label{eq wc 3 repetition BSC}
    \edited{\omegac(\s{XX'X''|AA'A'';BB'B''})_{P^{\alpha}\otimes P^{\alpha}\otimes P^{\alpha}} =
    \begin{cases}
    (1-\alpha)^6 & \textrm{ if } \alpha\in[0,\frac{1}{4}], \\
        \frac{1}{4}(1 - \alpha)^4(1+2\alpha) & \textrm{ if }\alpha\in(\frac{1}{4},\alpha_1], \\
        \frac{1}{8} & \textrm{ if } \alpha\in(\alpha_1,\frac{1}{2}],
    \end{cases}}
\end{equation}
\begin{equation}\label{eq wns 3 copies}
    \omegans(\s{XX'X''|AA'A'';BB'B''})_{P^{\alpha}\otimes P^{\alpha}\otimes P^{\alpha}} =
    \begin{cases}
        (1 - \alpha)^6 & \textrm{ if } \alpha\in [0,\frac{1}{4}] =: J_1, \\
        \frac{1}{4}(1-\alpha)^4(1+2\alpha)^2& \textrm{ if } \alpha\in (\frac{1}{4},\alpha_2] =: J_2, \\
        \frac{1}{7}(1-\alpha^3)^2 &\textrm{ if }  \alpha\in [\alpha_2,\alpha_3] =: J_3, \\
        \frac{1}{8} &\textrm{ if } \alpha\in[\alpha_3,\frac{1}{2}] =: J_4.
    \end{cases}
\end{equation}
\end{theorem}

\begin{proof}
   \edited{The proof is analogous to the proof of \cref{result w_c and w_ns repetition} for two parallel repetitions. In \file{BSC classical strategy n=3.nb} \cite{code} we perform an optimized search over all symmetric classical strategies leading to \eqref{eq wc 3 repetition BSC}.
   In \file{BSC no-signalling strategy n=3.nb} \cite{code} we provide explicit analytic $\alpha$-dependent solutions for the primal and dual linear programs for the no-signalling value. Both solutions have identical objective value that agrees with \eqref{eq wns 3 copies}.}
\end{proof}

\begin{figure}[ht]
\centering
\subfigure[]{\includegraphics[width=78mm]{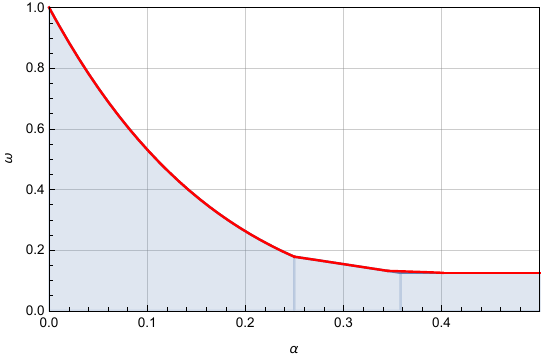}}
\subfigure[]{\includegraphics[width=78mm]{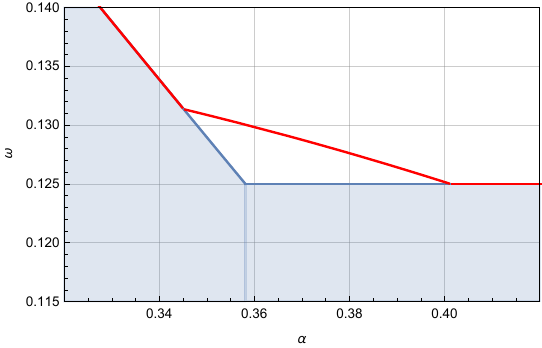}}
\caption{\edited{(a) Optimal classical (blue) and no-signalling (red) winning probabilities for the three-fold parallel repetition of the BSC game. The blue area represents the values below the optimal classical winning probabilities. (b) Zoom in of (a) for the values of $\alpha$ around $0.37$ where the classical and no-signalling values differ.}}
\label{fig:threeCopies}
\end{figure}

See \cref{fig:threeCopies} for a graphical representation of the optimal values from \cref{result 3 copies NS}. For `low' noise, $\alpha\in J_1$, the optimal value is attained by the classical strategy consisting on Alice and Bob outputting the received bit, i.e.~repeating three times the optimal classical strategy for a single copy of the game. On the other side, for `high' noise, $\alpha\in J_4$, the optimal value is attained by the classical strategy where Alice and Bob output a pre-agreed bit, which is also obtained by repeating the optimal strategy for a single copy. Therefore, 
\begin{equation}
    \omegac(\s{XX'X''|AA'A'';BB'B''})_{P^{\alpha}\otimes P^{\alpha}\otimes P^{\alpha}}=\omegans(\s{XX'X''|AA'A'';BB'B''})_{P^{\alpha}\otimes P^{\alpha}\otimes P^{\alpha}}=\omegac(\s{X|A;B})_{P^{\alpha}}^{3}, \textrm{ }\forall \alpha\in J_1\cup J_4.
\end{equation}
For $\alpha\in J_2$, the no-signalling optimal value can be attained by the deterministic strategy consisting on Alice and Bob outputting $111$ if they receive an input with more zeros than ones and outputting $000$ otherwise. See \cref{sec: arbitrary copies BSC} for no-signalling and classical strategies attaining this optimal value. For this interval, the optimal strategy for three copies is better than any combination of optimal two and one copies of the BSC game, i.e.
\begin{equation}
\begin{split}
    &\omegac(\s{XX'X''|AA'A'';BB'B''})_{P^{\alpha}\otimes P^{\alpha}\otimes P^{\alpha}}=\omegans(\s{XX'X''|AA'A'';BB'B''})_{P^{\alpha}\otimes P^{\alpha}\otimes P^{\alpha}}\\&>\omegans(\s{XX'|AA';BB'})_{P^{\alpha}\otimes P^{\alpha}}\omegans(\s{X''|A'';B''})_{P^{\alpha}}\geq \omegac(\s{XX'|AA';BB'})_{P^{\alpha}\otimes P^{\alpha}}\omegac(\s{X''|A'';B''})_{P^{\alpha}}\\&>\omegans(\s{X|A;B})_{P^{\alpha}}^{3}, \textrm{ }\forall\alpha\in J_2.
\end{split}
\end{equation}
For $\alpha\in J_3$ the following no-signalling strategy achieves the optimal value, as we explain in \cref{sec: arbitrary copies BSC}, 
\begin{align}\label{eq:three_copy_ns_strat}
    Q_3(x,y|a,b) =
    \begin{cases}
        \frac{1}{7} & \text{if } (x=y \text{ or } x\oplus b = 111 = y\oplus a) \text{ and } (x \oplus a \not= 111 \not= y\oplus b), \\
        0 & \text{otherwise.}
    \end{cases}
\end{align}

\subsection{Arbitrary parallel repetition}\label{sec: arbitrary copies BSC}

In this section, we will look to find classes of good strategies, both classical and no-signalling, for the $n$-fold parallel repetition of the BSC game.

\subsubsection{Classical strategies}\label{subsec:n_copy_class}

We have already seen some similarities in classical strategies between one, two and three copies of the game. For small $\alpha$, the best strategy is always to output the input (identity strategy). For $\alpha$ close to $1/2$ the best strategy is to output some fixed bitstring regardless of the input (constant strategy). The winning probabilities of these strategies for $n$ copies are $(1-\alpha)^{2n}$ and $2^{-n}$, respectively. For two and three copies, we also found similar strategies ``in between'' the identity and constant strategies. These strategies can also be extended to $n$ copies: outputting $0^n$ if the input contains at least as many zeros as ones and outputting $1^n$ otherwise (majority strategy). For odd $n$, the winning probability of the majority strategy is given by
\begin{equation}\label{for:majority_win_prob}
    \frac{1}{2^{n-1}}
    \left(\sum_{i=0}^{(n-1)/2} \binom{n}{i} \alpha^i (1-\alpha)^{n-i} \right)^2.
\end{equation}

An error-correcting code for the BSC consists of a message set $M$ and two functions $\Enc\colon M \to \{0,1\}^n$ and $\Dec\colon \{0,1\}^n \to M$. The objective of an error-correcting code is to send a message $m$ over the BSC by first encoding it using $\Enc$, sending the result over the BSC and recovering $m$ using $\Dec$, such that the probability of a correct recovery of $m$ is maximized. We will look at error-correcting codes more formally in \cref{sec:channel_games}. The readers already familiar with error-correcting codes will notice that the majority strategy is exactly applying $\Enc \circ \Dec$ from the repetition code to the input: the repetition code encodes messages 0 and 1 to $0^n$ and $1^n$ respectively and decodes by picking the bit that appears the most in the input. This motivates us to look at error-correcting codes to define strategies for $n$ repetitions of the BSC game.

\begin{example}
We consider the (7,4)-Hamming code, perhaps the most famous code for the BSC, introduced by Richard Hamming \cite{hamming1950error}. This code encodes bitstrings $d_1d_2d_3d_4$ of length 4 as bitstrings of length 7 by appending three parity bits: $d_1d_2d_3d_4p_1p_2p_3$. These bits represent the parity (XOR) of three of the original 4 bits (see \cref{fig:HamTikz}).

Decoding works by checking if the parity bits are still correct (still equal to the parity of the corresponding 3 bits). If this is the case, we just remove the last three bits of the received bitstring. Now suppose an error occurred in exactly one bit.
\begin{itemize}
    \item If the error occurred in $d_4$, all the parity bits are incorrect.
    \item If the error occurred in $d_1$, $d_2$ or $d_3$, two of the parity bits are incorrect ($p_1$ and $p_2$ for $d_1$, $p_1$ and $p_3$ for $d_2$ and $p_2$ and $p_3$ for $d_3$).
    \item If the error occurred in one of the parity bits, only that parity bit will be incorrect.
\end{itemize}
Using the above, we can perfectly deduce in which bit the error occurred and correct it accordingly. If more than one error occurs, this method never decodes correctly.

\begin{figure}[ht]
    \centering
    \begin{tikzpicture}
\def\size{\textwidth/5}
    
    \begin{scope}[transparency group]
    \begin{scope}[blend mode=screen]
    \node [circle, fill=blue, opacity=0.5, minimum size = \size] at (\size/2, 0) () {};
    \node [circle, fill=red, opacity=0.5, minimum size = \size] at (0, 0) () {};
    \node [circle, fill=green, opacity=0.5, minimum size = \size] at (\size/4, \size/2) () {};
    \end{scope}
    \end{scope}
    \node at (\size/4,3\size/4) {$p_1$};
    \node at (-\size/6,-\size/6) {$p_2$};
    \node at (4\size/6,-\size/6) {$p_3$};
    \node at (-\size/12,2\size/6) {$d_1$};
    \node at (7\size/12,2\size/6) {$d_2$};
    \node at (\size/4,-\size/6) {$d_3$};
    \node at (\size/4,\size/6) {$d_4$};
\end{tikzpicture}
    \caption{The Hamming code visualized: The bitstring $d_1d_2d_3d_4$ is encoded by appending the parity bits $p_1$, $p_2$ and $p_3$, where each parity bit represents the parity of the three bits inside their circle. A single error in one of the seven bits can be perfectly detected by checking which parity bits are incorrect.}
    \label{fig:HamTikz}
\end{figure}
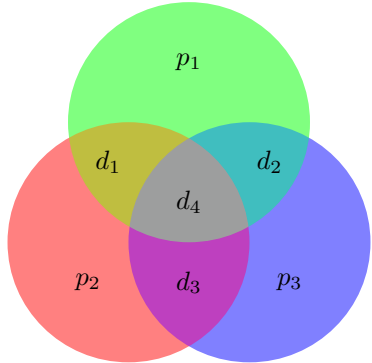

Since the Hamming code corrects exactly 0 or 1 error, we can write the average success probability of this code as
\[
    (1-\alpha)^7 + 7\alpha(1-\alpha)^6.
\]

Now consider the following strategy for 7 copies of the BSC game based on the Hamming code: both players perform the correction part of the Hamming code on their input and output the result (this is the same as decoding and then encoding again). It is obvious that the players win if and only if the initial bitstring $x$ is in the range of the encode function and the decoding of both players was successful. This observation results in the following winning probability:
\[
    \frac{2^4}{2^7} \left((1-\alpha)^7 + 7\alpha(1-\alpha)^6\right)^2.
\]

It turns out that this Hamming code strategy is strictly better for a large range of $\alpha$ than the identity, constant and majority strategy for 7 copies of the game. This confirms the idea that error-correcting codes define good classical strategies. 
\end{example}

\subsubsection{No-signalling strategies}\label{n_copies_section}

For two and three copies of the BSC game, we found the optimal no-signalling strategies~$Q_2$ and $Q_3$ (described in \cref{two_copy_ns_strat,eq:three_copy_ns_strat}). We can extend these no-signalling strategies to $n$ copies as follows:
\begin{equation}\label{for:no_signalling_strat}
    Q(x,y|a,b) =
    \begin{cases}
        \frac{1}{2^n - 1} & \text{if } (x=y \text{ or } x\oplus b = 1^n = y\oplus a) \text{ and } (x \oplus a \not= 1^n \not= y\oplus b), \\
        0 & \text{otherwise.}
    \end{cases}
\end{equation}

There is, however, a more intuitive way to describe this no-signalling correlation.
Alice outputs uniformly at random any bit string, except the negation of her input. Bob outputs the same string as Alice, except when that string happens to be the negation of his input, in which case he outputs the negation of Alice's input (see \cref{fig:list_strat}). Note that the roles of Alice and Bob in this description can be exchanged.

\begin{figure}[ht]
    \centering
    \begin{tikzpicture}
        \def\W{2.0cm}
        \def\H{0.5cm}
        \node at (-\W,\H) {Alice's input: 101};
        \node at ( \W,\H) {Bob's input: 001};
        \foreach \s/\i in {
            000/0,
            001/1,
            % 010/2,
            011/3,
            100/4,
            101/5,
            % 110/6,
            111/7} {
            \node (A\s) at (-\W,-\i*\H) {\s};
            \node (B\s) at ( \W,-\i*\H) {\s};
        }
        \node[red] (A010) at (-\W,-2*\H) {010};
        \node      (B010) at ( \W,-2*\H) {010};
        \node      (A110) at (-\W,-6*\H) {110};
        \node[red] (B110) at ( \W,-6*\H) {110};
        \foreach \s in {000,001,011,100,101,111} {
            \draw (A\s) -- (B\s);
        }
        \draw (A110.east) -- (B010.west);
    \end{tikzpicture}
    \caption{An example of a pairing of elements between the output sets of Alice and Bob, for three simultaneous copies. Each line represents a pair, and at the end of the process, one pair is chosen uniformly at random.\label{fig:list_strat}}
\end{figure}
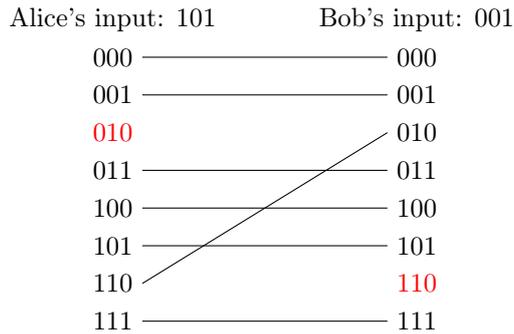

This formulation makes it obvious that we can define a more general class of no-signalling strategies: instead of the output sets consisting of everything apart from the opposite of the input, we can let the output sets consist of all bitstrings within Hamming distance $d$ from the input. We can then pair up the elements from the output sets and say that each of those pairs is output with equal probability. Again, if an element occurs in both lists, we pair it with itself. This description defines a no-signalling strategy, since Alice and Bob always output each of the elements of their output sets with the same probability, regardless of the input of the other. We denote by $Q_n^d$ a no-signalling strategy for $n$ copies of the BSC game defined by Hamming distance $d$. Note that for~$d \in \{1, \dots, n-2\}$ the strategy $Q_n^d$ is not unique, but they all achieve the same winning probability.

Let us find the winning probability of a strategy $Q_n^d$. Suppose that $x$ is the bitstring generated by the referee. The only way the players could output the combination $(x,x)$ is if both $d(x,a)\leq d$ and $d(x,b) \leq d$, in which case it is output with probability~$\left(\sum_{i=0}^d \binom{n}{d}\right)^{-1}$, since the sum is the size of their output sets. The probability that~$a$ lies within distance $d$ from $x$ is $\sum_{i=0}^d \binom{n}{i} \alpha^i(1-\alpha)^{n-i}$. We conclude that the winning probability of $Q_n^d$ is given by
\[
    \frac{1}{\sum_{i=0}^d \binom{n}{i}} \left(\sum_{i=0}^d \binom{n}{i} \alpha^i(1-\alpha)^{n-i}\right)^2.
\]

It turns out that all the optimal winning probabilities for one, two and three simultaneous copies of the BSC game can be achieved by a strategy of the form $Q_n^d$. If we pick~$d=0$ we get exactly the identity strategy. If we pick $d=n$, we get the average of all possible constant strategies (and by linearity, this achieves the same winning probability as a constant strategy). If we pick $d = n-1$, we get exactly the strategy defined in \cref{for:no_signalling_strat}. This strategy achieves winning probability
\begin{align*}
    \frac{1}{2^n - 1} \left(\sum_{i=0}^n \binom{n}{i} \alpha^i(1-\alpha)^{n-i} - \alpha^{n}\right)^2 = \frac{1}{2^n - 1} \left(1 - \alpha^{n}\right)^2.
\end{align*}

We are left with segment two for three copies. The strategy $Q_3^1$ achieves winning probability 
\[
    \frac{1}{4}\left((1-\alpha)^3 + 3\alpha(1-\alpha)^2\right)^2.
\]
This probability is exactly the same winning probability as the majority strategy, which we found to be optimal in this segment. We conclude that all optimal winning probabilities for one, two and three copies of the game can be achieved by a strategy of the form~$Q_n^d$.

It turns out that the class of strategies defined in this section can be described using a list-decoding scheme for the BSC channel. In the next section, we discuss strategies for a general channel $P_\s{A|X}$ based on error-correcting codes and list-decoding schemes.

\section{Channel LSSD games}\label{sec:channel_games}

In the previous section, we constructed an LSSD game based on a BSC. In this section, we extend this construction and define an LSSD game based on an arbitrary channel. For $n$ parallel instances of these games, we discuss classical strategies based on error-correcting codes and no-signalling strategies based on list-decoding schemes. We also investigate the asymptotic behaviour of the optimal winning probability as $n$ approaches infinity. Note that for any non-local game with optimal no-signalling winning probability smaller than 1 (and no promise on the input distribution), the optimal winning probability for $n$ parallel instances of the game exponentially goes to~$0$ \cite[Theorem 16]{buhrman_et_al}. Thus, we will be considering the limit of the exponent of the winning probability normalized by $n$.

We briefly recap basic concepts from information theory that we need in this section including entropic quantities and method of types. For a more in-depth introduction, see \cite[Chapter 2]{thomas2006elements} and \cite[Chapter~2]{csiszar2011information}.

Let $P$ be a probability distribution over $\SX$, and let $X$ be a random variable distributed according to $P$. We define the \emph{entropy} $H(X)_P = H(P)$ of $X$ as 
\[
    H(X)_P := -\sum_{x \in \SX} P(x)\log(P(x)),
\]
with the convention that $P(x)\log(P(x)) = 0$ wherever $P(x) = 0$. We drop subscript $P$ whenever the distribution of $X$ is clear from the context. Let $X$ and $Y$ be two random variables with joint probability distribution $P_\s{XY}$. The \emph{joint entropy} of $X$ and $Y$ is $H(X,Y)_P = H(P_\s{XY})$ and the \emph{conditional entropy} is
\[
    H(X|Y)_P := H(X,Y)_P - H(Y)_{P}.
\]
The \emph{mutual information} of two random variables $X$ and $Y$ is
\[
    I(X;Y)_P := H(X)_P + H(Y)_P - H(X,Y)_P.
\]
For two probability distributions $P$ and $Q$ over $\SX$, the \emph{relative entropy} is
\[
    D(P\|Q) := \sum_{x \in \SX} P(x) \log\left(\frac{P(x)}{Q(x)}\right).
\]  
If $P^1_\s{X|Y}$ and $P^2_\s{X|Y}$ are two conditional distributions over $\SX \times \SY$ and $Q_\s{Y}$ is a distribution over $\SY$, the corresponding \emph{conditional relative entropy} is
\[
    D(P^1_\s{X|Y} \| P^2_\s{X|Y} \mid Q_\s{Y}):= \sum_{y \in \SY} Q_\s{Y}(y) D(P^1_{\s{X}|\s{Y}=y} \| P^2_{\s{X}|\s{Y}=y}).
\]

We next introduce preliminaries on the \emph{method of types} (see \cite[Chapter~2]{csiszar2011information} for further reading). Let $\SX$ be a finite set and $n$ be a positive integer. For a sequence $x^n \in \SX^n$, its \emph{type} is a probability distribution $P$ over $\SX$ defined as
\begin{align}
    P(x) := \frac{\abs{\set{i: x_i = x}}}{n}.
\end{align}
Let $\mathcal{P}_n(\SX)$ denote the set of all types of sequences in $\SX^n$. For a given distribution $P$ over $\SX$, we denote by $\mathcal{T}_P$ all sequences in $\SX^n$ whose type is $P$. If $P_{\s{AX}}$ is a joint probability distribution on $\SA \times \SX$ and $x^n$ is a sequence in $\mathcal{T}_{P_{\s{X}}}$, we let $\mathcal{T}_{P_{\s{A|X}}}(x^n) := \set{a^n: (a^n,x^n) \in \mathcal{T}_{P_{\s{AX}}}}$. We need the following inequalities whose proofs are in \cite{csiszar2011information}:
\begin{align}
    |\mathcal{P}_n(\SX)| &\leq (n+1)^{|\SX|}, \\
    |\mathcal{T}_{P_{\s{X}}}| &\leq 2^{nH(X)_P} \label{eq:type2}, \\
    |\mathcal{T}_{P_{\s{A|X}}}(x^n)| &\geq \frac{2^{nH(A|X)_P}}{(n+1)^{|\SA|}}.\label{eq:type3}
\end{align}

We are now ready to define the main object of this section, channel LSSD games.

\begin{definition}
The channel LSSD game defined by $P_\s{A|X}$ is given by the probability distribution 
\[
    P_\s{XAB} = P_\s{X}P_\s{A|X}P_\s{B|X},
\]
with $P_\s{X}$ the uniform distribution over $\SX$, $\SA = \SB$, and $P_\s{B|X} = P_\s{A|X}$.
\end{definition}

Playing $n$ parallel copies of this channel game is the same as playing the channel game defined by the channel $P^{\times n}_\s{A|X}$, which can be thought of as the referee generating a string $x^n \in \SX^n$ and sending it to Alice and Bob by $n$ independent uses of their channels. Our main result of this section is the following characterization of the exponent of the optimal probability of winning for all three classes of strategies.

\begin{theorem}\label{the:classical_exponennt}
Let $P_\s{A|X}$ be a channel and let $P^{\times n}_{\s{XAB}}$ be the probability distribution defining the channel game corresponding to the channel $P^{\times n}_\s{A|X}$.
We have
\begin{align*}
        \lim_{n\to\infty} \frac{\log(\omegac(\s{X}^n|\s{A}^n;\s{B}^n)_{P^{\times n}})}{n} &=
    \lim_{n\to\infty} \frac{\log(\omegaq(\s{X}^n|\s{A}^n;\s{B}^n)_{P^{\times n}})}{n} = 
    \lim_{n\to\infty} \frac{\log(\omegans(\s{X}^n|\s{A}^n;\s{B}^n)_{P^{\times n}})}{n} \\
    &=\max_{Q_\s{XA}} I(X;A)_Q - 2D(Q_\s{A|X} \| P_\s{A|X}\mid Q_\s{X}) - \log(|\SX|).
\end{align*}
\end{theorem}

Note that to prove the theorem it is enough to prove the following two lemmas because of \cref{eq:pgrelations}.

\begin{lemma}[Achievability]
\label{lm:exponent-achv}
We have
\begin{align}
    \liminf_{n\to \infty}\frac{\log(\omegac(\s{X}^n|\s{A}^n;\s{B}^n)_{P^{\times n}})}{n} \geq \max_{Q_\s{XA}} I(X;A)_Q - 2D(Q_\s{A|X} \| P_\s{A|X}\mid Q_\s{X}) - \log(|\SX|).
\end{align}
\end{lemma}

\begin{lemma}[Converse]
\label{lm:exponent-converse}
We have
\begin{align}
    \limsup_{n\to \infty}\frac{\log(\omegans(\s{X}^n|\s{A}^n;\s{B}^n)_{P^{\times n}})}{n} \leq \max_{Q_\s{XA}} I(X;A)_Q - 2D(Q_\s{A|X} \| P_\s{A|X}\mid Q_\s{X}) - \log(|\SX|).
\end{align}
\end{lemma}

Our proof for these two lemmas is based on tools from information theory that we introduce in \cref{sec:tools-coding}.
We prove the first lemma in \cref{sec:achiev-exponent} by constructing a classical LSSD strategy from a code for the corresponding channel, and then by choosing an appropriate sequence of codes that optimize the winning probability. We prove the second lemma in \cref{sec:conv-exponent} by first relating the winning probability of an arbitrary no-signalling strategy to a list-decoding code, and then using a converse for list-decoding codes.

\subsection{Tools from information theory}\label{sec:tools-coding}

We recall here basic definitions concerning error-correction codes.  A code for $n$ uses of channel $ P_{\s{A|X}}$ operates as follows. The sender has a message set $M$ of possible messages. He picks one message $m \in M$  to send and encodes $m$ as a codeword $x^n$ of $\SX^n$, using a function $\Enc \colon M \to \SX^n$. Next, he transmits each of the symbols $x_i$ of this codeword to the receiver by consecutive uses of the channel; the receiver receives an $a^n$ in $\SA^n$ and decodes it to a message $m'$ using a function $\Dec \colon \SA^n \to M$. The communication was successful if $m = m'$. 

The minimum success probability of a code is given by 
\[
    \min_{m \in M} P^{\times n}_\s{A|X}(\Dec^{-1}(m)|\Enc(m)),
\]
and the rate of a code is $\frac{1}{n}\log|M|$. 

\begin{definition}\label{def:code}
We call a code $(\Enc, \Dec)$ for a channel $P_\s{A|X}$ an \emph{$(n,2^{nR},\alpha)$-code} if 
\begin{align}
    &\Enc \colon [2^{nR}] \to \SX^n,\\
    &\Dec\colon\SA^n\to[2^{nR}],\\
    &P^{\times n}_{\s{A|X}}(\Dec^{-1}(m)|\Enc(m)) \geq \alpha,
    \quad \forall m \in [2^{nR}].
\end{align}
\end{definition}

We know since Shannon's groundbreaking work \cite{Shannon1948} that there exists a sequence of codes with rate less than the capacity of the channel and probability of success tending to one. We also know from the strong-converse results \cite{Arimoto1973} that if the rate is above the capacity, the success probability exponentially tends to zero. In \cite{dueck1979reliability}, the optimal exponent of the success probability has been characterized. The following lemma is what we require in our achievability proof. Its proof resembles the proof in  \cite{dueck1979reliability}, but we need to modify it because we consider the minimum success probability, not the average. We leave the proof to \cref{app:A}.

\begin{restatable}{lemma}{achievability}\label{lem:achievable_code}
Let $P_\s{A|X}$ be a channel, $Q_\s{XA}$ a probability distribution over $\SX\times \SA$ and $\delta >0$. For \edited{$n\geq n_0(\card{\SX}, \card{\SA}, \delta)$}, there exists an 
\[
    \left(n, 2^{n(I(X;A)_Q - \delta)}, \frac{2^{-nD(Q_\s{A|X} \| P_\s{A|X}\mid Q_\s{X})}}{\edited{p(n)}}\right)
\]
code for the channel $P_\s{A|X}$ \edited{where $p(n)$ is a polynomial depending only on $\card{\SA}$}.
\end{restatable}

We next recall the definition of list decoding. The decoder here outputs a list of $L$ messages, instead of a single message. The decoding is successful if the list contains the correct message. We denote the list output by the decoder on input $a^n$ by $C_{a^n}$. The minimum success probability is then
\[
    \min_{m \in M} \sum_{a^n \in \SA^n :\; C_{a^n} \ni m}P_\s{A|X}^{\times n}(a^n|\Enc(m)). 
\]

\begin{definition}
We call a list-decoding code an $(n, 2^{nR}, L, \alpha)$-code if $\Dec$ maps elements $a^n$ of $\SA^n$ to subsets $C_{a^n}$ of $[2^{nR}]$ of size L and 
\begin{align*}
    &\Enc \colon [2^{nR}] \to \SX^n,\\
    &\sum_{a^n \in \SA^n :\; C_{a^n} \ni m}P^{\times n}_\s{A|X}(a^n|\Enc(m))\geq \alpha,
    \quad \forall m \in [2^{nR}].
\end{align*}
\end{definition}

We have the following converse for list-decoding schemes, see \cref{app:proof-list-decode-conv} for the proof.

\begin{restatable}{lemma}{converse}\label{lem:list-decoding_upperbound}
{For any  list-decoding $(n, 2^{nR}, 2^{nR_L}, 2^{-n\zeta_n})$ code for $P_\s{A|X}$, we have
\[
    \zeta_n \geq \min_{Q_\s{XA}}\left[D(Q_\s{A|X} \| P_\s{A|X}\mid Q_\s{X}) + \max \{R -R_L- I(X;A)_Q,0\}\right] + \edited{O\pr*{\frac{\log n}{n}}},
\]}
\edited{where the constant hidden in $O(\cdot)$ depends only on $\card{\SX}$ and $\card{\SA}$.}
\end{restatable}

\subsection{Achievability: Classical strategies from error-correction codes}\label{sec:achiev-exponent}

We prove \cref{lm:exponent-achv} in this section, which we re-state for readers' convenience.

\begin{lemma} 
We have
\begin{align}
    \liminf_{n\to \infty}\frac{\log(\omegac(\s{X}^n|\s{A}^n;\s{B}^n)_{P^{\times n}})}{n} \geq \max_{Q_\s{XA}} I(X;A)_Q - 2D(Q_\s{A|X} \| P_\s{A|X}\mid Q_\s{X}) - \log(|\SX|).
\end{align}
\end{lemma}

\begin{proof}
We first explain how to use error-correction codes to find classical strategies for the parallel repetition of a channel LSSD game. Let $(\Enc, \Dec)$ be an $(n, 2^{nR}, \alpha)$-code for the channel $P_\s{A|X}$.
We consider a classical LSSD strategy for $n$ parallel repetitions of the channel LSSD game in which both players use the estimation function $f := \Enc \circ \Dec$.
%\mo{Should it be $\Dec \circ \Enc$?}\cs{No.}
This strategy can be interpreted as the players decoding directly to the codeword of a message instead of to the message itself.  We lower bound  the winning probability of the strategy given by $f$ as
\begin{align}
    \frac{1}{|\SX|^n} \sum_{x^n \in \SX^n} P_\s{A|X}^{\times n}(f^{-1}(x^n) | x^n)^2 &= \frac{1}{|\SX|^n} \sum_{x^n \in \SX^n} \left(\sum_{a^n \in \SA^n : f(a^n) = x^n}P_\s{A|X}^{\times n}(a^n | x^n)\right)^2 \label{for:win_the5.1}\\
%    &= \sum_{x^n \in \im(\Enc)} \left(\sum_{a^n \in \SA^n: f(a^n) = x^n}P^{\times n}(a^n | x^n)\right)^2\\
    &\overset{(a)}{\geq} \frac{1}{|\SX|^n} \sum_{x^n \in \im(\Enc)}\alpha^2 = \frac{2^{nR}}{|\SX|^n}\alpha^2,\label{for:lower_bound}
\end{align}
where $(a)$ follows since $(\Enc, \Dec)$ is an $(n, 2^{nR}, \alpha)$-code. Notice that there is a trade-off between the success probability and the number of messages. We simultaneously want the success probability and the number of messages to be large. However, increasing one necessarily means decreasing the other.

Let $Q_\s{XA}$ be a distribution on $\SX \times \SA$ and $\delta>0$. \edited{Let $n_0(\card{\SX}, \card{\SA}, \delta)$ and $p(n)$ be as in \cref{lem:achievable_code} where $p(n)$ is a polynomial only depending on $\card{\SA}$.} By \cref{lem:achievable_code}, for \edited{$n\geq n_0(\card{\SX}, \card{\SA}, \delta)$}, there exists an $\left(n, 2^{n(I(X;A)_Q - \delta)}, \frac{2^{-nD(Q_\s{A|X} \| P_\s{A|X}\mid Q_\s{X})}}{\edited{p(n)}}\right)$-code for $P_\s{A|X}$. Let $f = \Enc \circ \Dec$ be the strategy defined by this code. The winning probability of this strategy is at most the optimal classical winning probability, so by using \cref{for:lower_bound} we find
\begin{align*}
    \omegac(\s{X}^n|\s{A}^n;\s{B}^n)_{P^{\times n}} \geq \frac{2^{n(I(X;A)_Q - \delta -D(Q_\s{A|X}\|P_\s{A|X}\mid Q_\s{X}))}}{|\SX|^n \edited{p(n)}},
\end{align*}
and therefore
\begin{align}\label{for:before_limit}
    \frac{\log(\omegac(\s{X}^n|\s{A}^n;\s{B}^n)_{P^{\times n}})}{n} \geq I(X;A)_Q -\delta - 2D(Q_\s{A|X}\|P_\s{A|X}\mid Q_\s{X}) - \log(|\SX|) - \frac{\log(\edited{p(n)})}{n}.
\end{align}
Since \cref{for:before_limit} holds for any $Q_\s{XA}$ and $\delta >0$, and $\lim_{n \to \infty} \frac{\log(\edited{p(n)})}{n} = 0$, we conclude that
\[
    \lim_{n\to\infty} \frac{\log(\omegac(\s{X}^n|\s{A}^n;\s{B}^n)_{P^{\times n}})}{n} \geq \max_{Q_\s{XA}} I(X;A)_Q - 2D(Q_\s{A|X} \| P_\s{A|X}\mid Q_\s{X}) - \log(|\SX|).
\]
This completes the proof of the achievability of the error exponent.
\end{proof}

\subsection{Converse: No-signalling LSSD strategies and list-decoding codes}\label{sec:conv-exponent}

We prove \cref{lm:exponent-converse} in this section which we restate for readers' convenience.

\begin{lemma}\label{lem:converse}
We have
\begin{align}
    \limsup_{n\to \infty}\frac{\log(\omegans(\s{X}^n|\s{A}^n;\s{B}^n)_{P^{\times n}})}{n} \leq \max_{Q_\s{XA}} I(X;A)_Q - 2D(Q_\s{A|X} \| P_\s{A|X}\mid Q_\s{X}) - \log(|\SX|).
\end{align}
\end{lemma}

We first prove the existence of an optimal strategy invariant under permutations of inputs and outputs. To this end, we need the following notation: For a permutation $\sigma \in S_n$ and a sequence $x^n \in \SX^n$ we denote by $\sigma(x^n) \in \SX^n$ the sequence obtained from $x^n$ by permuting its entries according to $\sigma$.

\begin{lemma}\label{lem:permutations} For $n$ parallel repetitions of a channel LSSD game, there is  an optimal no-signalling strategy $Q$ such that 
\begin{equation}\label{for:second_property}
    \forall \sigma \in S_n: \;\;\; Q(\sigma(x^n), \sigma(y^n) | \sigma(a^n), \sigma(b^n)) = Q(x^n,y^n|a^n,b^n).
\end{equation}
\end{lemma}

\begin{proof}
Let $Q$ be an optimal strategy and $\sigma \in S_n$. The strategy $Q_\sigma$ defined by $Q_\sigma(x^n,y^n| a^n,b^n) = Q(\sigma(x^n) , \sigma(y^n)  | \sigma(a^n), \sigma(b^n))$ has the same winning probability as $Q$, since the $n$-fold probability distribution is invariant under permutations: $P^{\times n}_\s{XAB}(x^n,a^n,b^n) = P^{\times n}_\s{XAB}(\sigma(x^n),\sigma(a^n),\sigma(b^n))$. We define 
\[
    \hat{Q} := \frac{1}{n!}\sum_{\sigma \in S_n}Q_\sigma.
\]
The strategy $\hat{Q}$ satisfies \eqref{for:second_property}: for any $\tau \in S_n$,
 \begin{align*}
     \hat{Q}(\tau(x^n) , \tau(y^n) | \tau(a^n), \tau(b^n)) &= \frac{1}{n!}\sum_{\sigma \in S_n}Q_\sigma(\tau(x^n), \tau(y^n) | \tau(a^n), \tau(b^n))\\
     &= \frac{1}{n!}\sum_{\sigma \in S_n}Q(\sigma(\tau(x^n)) , \sigma(\tau(y^n)) | \sigma(\tau(a^n)), \sigma(\tau(b^n)))\\
     &= \frac{1}{n!}\sum_{\pi \in S_n}Q(\pi(x^n), \pi(y^n)| \pi(a^n), \pi(b^n))\\
     &= \frac{1}{n!}\sum_{\pi \in S_n}Q_\pi(x^n, y^n | a^n, b^n)\\
     &= \hat{Q}(x^n,y^n|a^n,b^n).
 \end{align*}
 Finally, by linearity of the winning probability, $\hat{Q}$ also achieves the same winning probability as $Q$, which means that it is optimal.
\end{proof}

\begin{proof}[Proof of \cref{lem:converse}]
Let $Q$ be an optimal strategy satisfying \eqref{for:second_property}. Its marginal distributions $Q(x^n|a^n)$ and $Q(y^n|b^n)$ only depend on the joint type of $(x^n, a^n)$ and $(y^n,b^n)$, respectively. In particular, we can write the winning probability of $Q$ as follows:
\begin{align}
    \omegans(\s{X}^n|\s{A}^n;\s{B}^n)_{P^{\times n}} &= \sum_{x^n,a^n,b^n} P^{\times n}_\s{XAB}(x^n,a^n,b^n) Q(x^n,x^n|a^n,b^n) \nonumber \\
    &= \sum_{\substack{Q_{\s{XA}}\in \mathcal{P}_n(\SX \times \SA) \\ Q'_{\s{XB}}\in \mathcal{P}_n(\SX \times \SB)}}
    \sum_{\substack{x^n,a^n,b^n:\\
    (x^n,a^n) \in \mathcal{T}_{Q_\s{XA}}\\
    (x^n,b^n) \in \mathcal{T}_{Q'_\s{XA}}}} P^{\times n}_\s{XAB}(x^n,a^n,b^n)Q(x^n,x^n|a^n,b^n), \label{for:weight_win_prob}
\end{align}
where $\mathcal{P}_n(\cdot)$ denotes the set of types of length-$n$ strings over a given set, and $\mathcal{T}_{\cdot}$ denotes all sequences of a given type.

Since there are $(n + 1)^{2|\SX||\SA|}$ terms in the first sum of \eqref{for:weight_win_prob}, there must exist $Q_{\s{XA}}\in \mathcal{P}_n(\SX \times \SA)$ and  $Q'_{\s{XB}}\in \mathcal{P}_n(\SX \times \SB)$ such that 
\begin{equation}\label{for:ij}
    \sum_{\substack{x^n,a^n,b^n:\\
    (x^n,a^n) \in \mathcal{T}_{Q_\s{XA}}\\
    (x^n,b^n) \in \mathcal{T}_{Q'_\s{XA}}}} P^{\times n}_\s{XAB}(x^n,a^n,b^n)Q(x^n,x^n|a^n,b^n) \geq \frac{\omegans(\s{X}^n|\s{A}^n;\s{B}^n)_{P^{\times n}}}{(n+1)^{2|\SX||\SA|}}.
\end{equation}

Let us define for each $a^n$ and $b^n$
\begin{align}
    C_{a^n} := \{x^n: (x^n, a^n)\in\mathcal{T}_{Q_\s{XA}} \},\\
    D_{b^n} := \{x^n: (x^n, b^n)\in\mathcal{T}_{Q'_\s{XB}} \}.
\end{align}
Now consider the following strategy $\tilde{Q}$:
\begin{itemize}
    \item on input $(a^n,b^n)$, Alice and Bob generate $(x^n,y^n)$ according to $Q$;
    \item Alice checks if $x^n \in C_{a^n}$ and if not, uniformly generates a new output $\tilde{x}^n$ from $C_{a^n}$ (if $C_{a^n}$ is the empty set, Alice generates an arbitrary output);
    \item Bob checks if $y^n \in D_{b^n}$ and if not, uniformly generates a new output $\tilde{y}^n$ from $D_{b^n}$ (if $D_{b^n}$ is the empty set, Bob generates an arbitrary output).
\end{itemize}
This strategy is no-signalling and has winning probability of at least $\frac{\omegans(\s{X}^n|\s{A}^n;\s{B}^n)_{P^{\times n}}}{(n+1)^{2|\SX||\SA|}}$, by \eqref{for:ij}. We also have that $\tilde{Q}(x^n|a^n)$ is uniform over $C_{a^n}$ when $C_{a^n}\neq \0$, since $Q(x^n|a^n)$ only depends on the joint type of $(x^n,a^n)$. Similarly, $\tilde{Q}(y^n|b^n)$ is uniform over $D_{b^n}$ when $D_{b^n}\neq \0$. 

Note that for any $a^n$ and ${a'}^n$, if $C_{a^n}$ and $C_{{a'}^n}$ are non-empty, then $|C_{a^n}| = |C_{{a'}^n}|$. We define $L_A := |C_{a^n}|$ for a non-empty $C_{a^n}$ and similarly define $L_B := |D_{b^n}|$ for a non-empty $D_{b^n}$. We find 
\begin{align*}
    \frac{\omegans(\s{X}^n|\s{A}^n;\s{B}^n)_{P^{\times n}}}{(n+1)^{2|\SX||\SA|}} &\leq \sum_{x^n,a^n,b^n} P^{\times n}_\s{XAB}(x^n,a^n,b^n)\tilde{Q}(x^n,x^n|a^n,b^n)\\
    &\leq \sum_{x^n,a^n,b^n} P^{\times n}_\s{XAB}(x^n,a^n,b^n)\min\{\tilde{Q}(x^n|a^n) , \tilde{Q}(x^n|b^n)\}\\
    &\leq \frac{1}{\max\{L_A, L_B\}} \sum_{x^n,a^n,b^n} P^{\times n}_\s{XAB}(x^n,a^n,b^n) \delta(x^n \in C_{a^n}) \delta(x^n \in D_{b^n})\\
    &= \frac{1}{\max\{L_A, L_B\} |\SX|^n} \sum_{x^n} \left(\sum_{a^n: \; C_{a^n}\ni x^n} P^{\times n}_\s{A|X}(a^n|x^n) \right)\left(\sum_{b^n:\; D_{b^n} \ni x^n} P^{\times n}_\s{A|X}(b^n|x^n) \right).
\end{align*}
\edited{
Upon defining 
\begin{align*}
    q_A(x^n) &:= \sum_{a^n:\; C_{a^n} \ni x^n} P^{\times n}_\s{A|X}(a^n|x^n),\\
    q_B(x^n) &:= \sum_{b^n:\; D_{b^n}\ni x^n} P^{\times n}_\s{A|X}(b^n|x^n),
\end{align*}
we can write 
\begin{align}
    \sum_{x^n} \left(\sum_{a^n: \; C_{a^n}\ni x^n} P^{\times n}_\s{A|X}(a^n|x^n) \right)\left(\sum_{b^n:\; D_{b^n} \ni x^n} P^{\times n}_\s{A|X}(b^n|x^n) \right) = \sum_{x^n} q_A(x^n) q_B(x^n).
\end{align}
By Cauchy-Schwartz inequality, we have
\begin{align}
\sum_{x^n} q_A(x^n) q_B(x^n) \leq \sqrt{\pr*{\sum_{x^n}p_A(x^n)^2}\pr*{\sum_{x^n}p_B(x^n)^2}}.
\end{align}
Therefore, without loss of generality, we can assume that
\begin{align}
    \sum_{x^n} q_A(x^n) q_B(x^n) \leq \sum_{x^n}p_A(x^n)^2.
\end{align}
We can upper-bound the winning probability of the strategy as
\begin{align}
    \frac{\omegans(\s{X}^n|\s{A}^n;\s{B}^n)_{P^{\times n}}}{(n+1)^{2|\SX||\SA|}} \leq \frac{1}{\max(L_A, L_B)|\SX|^n} \sum_{x^n}  q_A(x^n)^2\leq \frac{1}{L_A|\SX|^n} \sum_{x^n}  q_A(x^n)^2. \label{eq:ns-vs-list}
\end{align}}

Let $\delta > 0$. For each $i \geq 0$, we define
\[
    \mathcal{R}_i := \{x^n \in \SX^n \;|\; 2^{-n\delta(i+1)} \leq q_A(x^n) < 2^{-n\delta i}\}.
\]
We define a list-decoding scheme $(\Enc_i, \Dec_i)$ as follows: $\Enc_i \colon \mathcal{R}_i \to \SX^n$ is the identity function and 
\[
    \Dec_i(a^n) = C_{a^n} \cap \mathcal{R}_i.
\]
Note that intersecting $C_{a^n}$ with $\mathcal{R}_i$ only decreases the size of the list, making the code weaker. This observation means that we will still be able to use \cref{lem:list-decoding_upperbound} for a list decoding with list size $L$. For each $x^n \in \mathcal{R}_i$, we have 
\[
    \sum_{a^n :\; \Dec_i(a^n) \ni x^n} P^{\times n}_\s{A|X}(a^n|x^n) \geq q_A(x^n) \geq 2^{-n\delta(i+1)},
\]
so $(\Enc_i, \Dec_i)$ defines an $(n, |\mathcal{R}_i|, L, 2^{-\delta(i+1)})$-code. By \cref{lem:list-decoding_upperbound}, we have 
\[
    \delta(i+1) \geq \min_{Q_\s{XA}}D(Q_\s{A|X} \| P_\s{A|X}\mid Q_\s{X}) + \max \left\{\frac{\log|\mathcal{R}_i|}{n} - \frac{\log(L_A)}{n}- I(X;A)_Q,0\right\}+O\pr*{\frac{\log n}{n}}.
\]
We find that if $q_A(x^n) >0$, then $q_A(x^n) \geq 2^{-n\mu}$, with $\mu := \max_{x,a:P_{\s{A|X}}(a|x) > 0}-\log(P_{\s{A|X}}(a|x))$. Thus, if $i \geq t := \lfloor \frac{\mu}{\delta}\rfloor$, then $\mathcal{R}_i$ is empty. Now, we find
\begin{align}
    \frac{1}{L_A}\sum_{x^n \in \SX^n} q_A(x^n)^2 &= \sum_{i=0}^t \sum_{x^n \in \mathcal{R}_i}\frac{1}{L_A}q_A(x^n)^2\\
    &\leq \sum_{i=0}^t \frac{|\mathcal{R}_i|}{L_A}2^{-2n\delta i }\\
    & \leq \sum_{i=0}^t 2^{n \left(\frac{\log|\mathcal{R}_i|}{n} - \frac{\log(L_A)}{n} -2\min_{Q_\s{XA}}\left(D(Q_\s{A|X} \| P_\s{A|X}\mid Q_\s{X}) + \max \left\{\frac{\log|\mathcal{R}_i|}{n} -\frac{\log(L)}{n}- I(X;A)_Q,0\right\}\right) +O\pr*{\frac{\log n}{n}} + \delta\right)}\\
    &\leq \sum_{i=0}^t 2^{n \left(\max_{Q_\s{XA}}\left(I(X;A)_Q -2D(Q_\s{A|X} \| P_\s{A|X} \mid Q_\s{X})\right) +O\pr*{\frac{\log n}{n}} + \delta\right)}\\
    &= \left(\lfloor \frac{\mu}{\delta}\rfloor+1\right) 2^{n \left(\max_{Q_\s{XA}}\left(I(X;A)_Q -2D(Q_\s{A|X} \| P_\s{A|X} \mid Q_\s{X})\right) +O\pr*{\frac{\log n}{n}} + \delta\right)}. \label{eq:ns-list-up}
\end{align}
Combining \eqref{eq:ns-vs-list} and \eqref{eq:ns-list-up}, taking logarithm from both sides, and choosing $\delta = 1/n$ yields the desired converse results.
\end{proof}

\subsection{Calculating the exponent for BSCs}

We calculate, for  BSCs, the value of the limit of the exponent in \cref{the:classical_exponennt}: $\max_{Q_\s{XA}} I(X;A)_Q - 2D(Q_\s{A|X} \| P_\s{A|X}\mid Q_\s{X}) - \log(|\SX|)$. To this extent, let $Q_\s{XA}$ be a distribution over $\{0,1\}\times \{0,1\}$.
Let us calculate the exponent one step at a time. First of all, we have 
\[
    I(X;A)_Q = H(X)_Q + H(A)_Q - H(X,A)_Q
\]
and
\[
    H(X)_Q = -\sum_{x=0}^1 Q_\s{X}(x) \log(Q_\s{X}(x)) = -\sum_{x=0}^1 \left(\sum_{a=0}^1 Q_\s{XA}(x,a)\right)\log\left(\sum_{a=0}^1 Q_\s{XA}(x,a)\right).
\]
We can find $H(A)_Q$ in a similar way. We also have
\[
    H(X,A)_Q = -\sum_{x,a=0}^1 Q_\s{XA}(x,a)\log(Q_\s{XA}(x,a)).
\]

Now let us find the value of $D(Q_\s{A|X}\|P_\s{A|X}\mid Q_\s{X})$:
\begin{align*}
    D(Q_\s{A|X}\|P_\s{A|X}\mid Q_\s{X}) &= \sum_{x=0}^1 Q_\s{X}(x) D(Q_{\s{A|X}=x} \| P_{\s{A|X}=x})\\
    &= \sum_{x=0}^1 \left(\sum_{a=0}^1 Q_\s{XA}(x,a)\right)\left(\sum_{a=0}^1 Q_\s{A|X}(a|x) \log\left(\frac{Q_\s{A|X}(a|x)}{P_\s{A|X}(a|x)}\right)\right).
\end{align*}
Using numerical analysis we found that the maximum $\max_{Q_\s{XA}} I(X;A)_Q - 2D(Q_\s{A|X} \| P_\s{A|X} \mid Q_\s{X}) - \log(|\SX|)$ is always achieved by a distribution $Q_\s{XA}$ for which $Q_\s{XA}(0,0) = Q_\s{XA}(1,1) =:c$ and $Q_\s{XA}(0,1) = Q_\s{XA}(1,0) =:d$. Using this property, we have 
\[
    H(X)_Q = H(A)_Q = -2(c+d)\log(c+d),
\]
and 
\[
    H(X,Y)_Q = -2c\log(c) - 2d\log(d).
\]
We also find
\[
    D(Q_\s{A|X}\|P_\s{X|A}\mid Q_\s{X}) = 2\left(c\log\left(\frac{c}{(c+d)(1-\alpha)}\right) + d\log\left(\frac{d}{(c+d)\alpha}\right)\right).
\]

Combining the expressions above, we find the value $I(X;A)_Q - 2D(Q_\s{A|X} \| P_\s{A|X}\mid Q_\s{X}) - \log(|\SX|)$. Note that for $Q_\s{XA}$ to be a distribution, we need $d=\frac{1}{2} - c$. This observation means that we only need to maximize with respect to the variable $c$ (we see $\alpha$ as a constant). We can solve this maximization by calculating the derivative, setting it to 0 and solving for $c$. Using a computer algebra system, we find 
\begin{equation}\label{for:limit}
    \max_{Q_\s{XA}} I(X;A)_Q - 2D(Q_\s{A|X} \| P_\s{A|X} \mid Q_\s{X}) - \log(|\SX|) = \log(1 - 2(1-\alpha)\alpha).
\end{equation}

In \cref{fig:pretty} we plotted this expression together with the exponent of the optimal winning probability achieved by the strategies $Q_n^d$ for some $n$ (see \cref{n_copies_section}). We can clearly see how this exponent approaches the limit calculated in \eqref{for:limit}.

\begin{figure}[ht]
    \centering
    \includegraphics[width=0.9\textwidth]{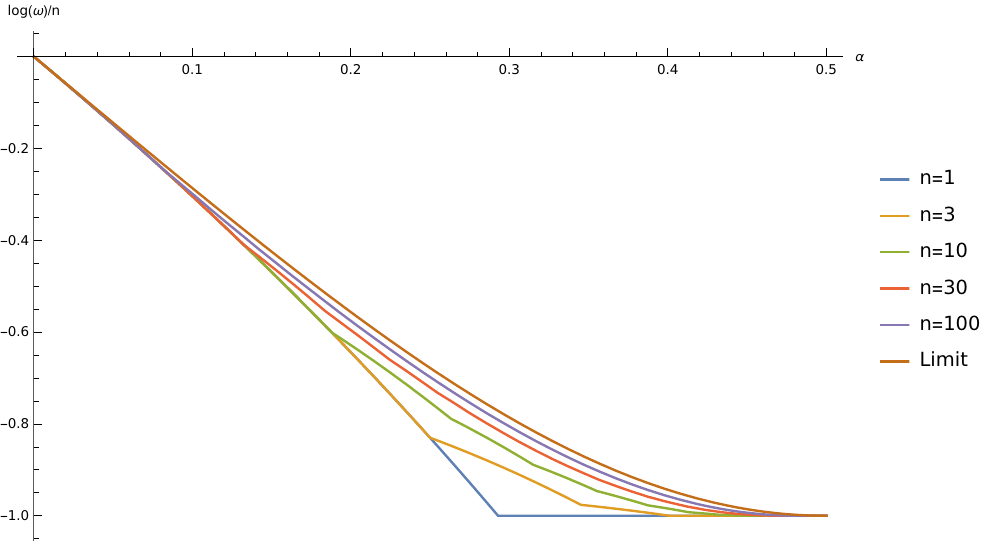}
    \caption{Plot of $\log(\omega)/n$ for different values of $n$ and the limit of this expression, given by \cref{for:limit}, against $\alpha$. We calculated $\omega$ as the optimal winning probability achieved by the strategies $Q_n^d$ (see \cref{n_copies_section}).}
    \label{fig:pretty}
\end{figure}

\subsection*{Acknowledgements}

We thank Christian Majenz for useful discussions.
MO was supported by an NWO Vidi grant (Project No.~VI.Vidi.192.109).
LEF was supported by the Dutch Ministry of Economic Affairs and Climate Policy (EZK), as part of the Quantum Delta NL programme.

\bibliographystyle{quantum}
\bibliography{references}

\appendix

\section{Proofs}

\subsection{Proof of \cref{the:symmetry}}\label{sec: proot theorem symmetry}

Let two functions $f \colon \SA \to \SX$ and $g \colon \SB \to \SX$ define a deterministic strategy.
We prove that either Alice and Bob both performing $f$ or both performing $g$ can only increase the winning probability. Note that Alice and Bob can perform the same strategy, since $\SA = \SB$. The winning probability of the strategy defined by $f$ and $g$ is given by 
\begin{align*}
    &\sum_{\substack{x\in\SX\\a\in\SA, b\in\SB}}P_\s{XAB}(x, a, b) \delta[f(a) = g(b) = x]\\ 
    & {=} \frac{1}{|\SX|} \sum_{\substack{x\in\SX\\a\in\SA, b\in\SB}}P_\s{AB|X}(a,b|x) \delta[f(a) = g(b) = x]\\ 
    &{=} \frac{1}{|\SX|} \sum_{x\in\SX} \left(\sum_{a \in \SA} P_\s{A|X}(a | x)\delta[f(a) = x]\right) \left(\sum_{b \in \SB} P_\s{B|X}(b | x)\delta[g(b) = x]\right)\\
    &{=} \frac{1}{|\SX|} \sum_{x\in\SX} P_\s{A|X}(f^{-1}(x)|x)P_\s{A|X}(g^{-1}(x)|x),
\end{align*}
where in the first, second and third equalities we have used hypotheses (i), (ii) and (iii) of \cref{the:symmetry}, respectively and notice that $f^{-1}(x)$ and $g^{-1}(x)$ might be sets.
Write $q_f(x) := P_\s{A|X}(f^{-1}(x)|x)$ and $q_g(x) := P_\s{A|X}(g^{-1}(x)|x)$. Notice that $q_f$ and $q_g$ are vectors indexed by $x \in \SX$, so we can write the winning probability as an inner product of these vectors: 
\begin{equation}\label{for:win_prob}
    \frac{1}{|\SX|} \left<q_f, q_g\right>.
\end{equation}
Using the Cauchy--Schwarz inequality,
\[
    |\left<q_f, q_g\right>|^2 \leq \left<q_f, q_f\right>\left<q_g,q_g\right>,
\]
and thus we cannot have $\left<q_f, q_g\right> > \left<q_f, q_f\right>$ and 
$\left<q_f, q_g\right> > \left<q_g, q_g\right>$. Therefore, we can conclude that Alice and Bob either both performing $f$ or both performing $g$ does not decrease the winning probability given in \cref{for:win_prob}. Now suppose we picked $f$ and $g$ to form an optimal strategy, then by the previous statement, we immediately find a symmetric deterministic strategy that is also optimal.

\subsection{Proof of \cref{lem:achievable_code}}\label{app:A}

The proof of \cref{lem:achievable_code} relies on concepts and theorems from the book by Csiszár and Körner \cite{csiszar2011information}. We will not be discussing these concepts here. We repeat the statement of \cref{lem:achievable_code} here for the reader's convenience.

\achievability*

% \begingroup
% \def\thetheorem{\ref{lem:achievable_code}}
% \begin{lemma}
% Let $P_\s{A|X}$ be a channel, $Q_\s{XA}$ a probability distribution over $\SX\times \SA$ and $\delta >0$. For \edited{$n\geq n_0(\card{\SX}, \card{\SA}, \delta)$}, there exists an 
% \[
%     \left(n, 2^{n(I(X;A)_Q - \delta)}, \frac{2^{-nD(Q_\s{A|X} \| P_\s{A|X}\mid Q_\s{X})}}{\edited{p(n)}}\right)
% \]
% code for the channel $P_\s{A|X}$ \edited{where $p(n)$ is a polynomial depending only on $\card{\SA}$}.
% \end{lemma}
% \addtocounter{theorem}{-1}
% \endgroup

\begin{proof}
Let $R = I(X;A)_Q - \delta$. By the packing lemma (Lemma~10.1 in \cite{csiszar2011information}), there exists a function $\Enc \colon [2^{nR}] \to \SX^n$ such that 
\begin{itemize}
    \item $\Enc(m)$ is of type $Q_X$ for all $m \in [2^{nR}]$;
    \item $|\edited{\mathcal{T}}_{Q_{A|X}}(\Enc(m)) \cap \bigcup_{m' \not=m} \edited{\mathcal{T}}_{Q_{A|X}}(\Enc(m'))| \leq |\edited{\mathcal{T}}_{Q_{A|X}}(\Enc(m))|2^{-n\frac{\delta}{2}}$
\end{itemize}
(Note that the conditions of the packing lemma are satisfied, because $H(X)_Q \geq I(X;A)_Q$).

Now define $\Dec \colon \SA^n \to [2^{nR}]$ by $\Dec(a^n) = m$ if $m$ is the unique message such that~$a^n \in \edited{\mathcal{T}}_{Q_{A|X}}(\Enc(m))$, otherwise we set $\Dec(a^n) = 0$. For all $m \in [2^{nR}]$, we have 
\begin{equation}\label{for:name}
    \sum_{a^n :\; \Dec(a^n) = m} P^{\times n}_\s{A|X}(a^n | \Enc(m)) = |\Dec^{-1}(m)| 2^{-n(D(Q_\s{A|X} \| P_\s{A|X} | Q_\s{X}) + H(A|X)_Q)}
\end{equation}
by Lemma~2.6 in \cite{csiszar2011information} (using that $\Enc(m)$ are all of type $Q_\s{X}$). By definition of the decoder, we also have  
\begin{align}
    |\Dec^{-1}(m)| & \geq |\edited{\mathcal{T}}_{Q_{A|X}}(\Enc(m)) \setminus \bigcup_{m' \not=m} \edited{\mathcal{T}}_{Q_{A|X}}(\Enc(m'))|\\
    & \geq |\edited{\mathcal{T}}_{Q_{\s{A|X}}}(\Enc(m))|(1-2^{-n\frac{\delta}{2}}) \label{for:by_enc}\\
    &\geq (n + 1)^{-|\SA|}(1-2^{-n\frac{\delta}{2}}) 2^{nH(A|X)_Q} \label{for:by_lem}
\end{align}
where \eqref{for:by_enc} follows from the second property of $\Enc$ and \eqref{for:by_lem} follows from \eqref{eq:type3}. By combining \eqref{for:by_lem} with \eqref{for:name} we conclude that $(\Enc, \Dec)$ is a \[
   \left(n, 2^{n(I(X;A)_Q - \delta)}, (n + 1)^{-|\SA|}(1-2^{-n\frac{\delta}{2}}){2^{-nD(Q_\s{A|X} \| P_\s{A|X}|Q_\s{X})}}\right) 
\]
code. \edited{We finally choose $p(n) = 2(n+1)^{\abs{\SA}}$ which is a polynomial in $n$ depending only on $\card{\SA}$ and for $n\geq \frac{2}{\delta}$ we have  $p(n)^{-1} \leq (n + 1)^{-|\SA|}(1-2^{-n\frac{\delta}{2}})$. This concludes the existence of an \[\left(n, 2^{n(I(X;A)_Q - \delta)}, \frac{{2^{-nD(Q_\s{A|X} \| P_\s{A|X}|Q_\s{X})}}}{p(n)}\right)\] code.}
\end{proof}

\subsection{Proof of \cref{lem:list-decoding_upperbound}}\label{app:proof-list-decode-conv}
We first repeat the statement of \cref{lem:list-decoding_upperbound} for the reader's convenience.

\converse*

% \begingroup
% \def\thetheorem{\ref{lem:list-decoding_upperbound}}
% \begin{lemma}
% For any  list-decoding $(n, 2^{nR}, 2^{nR_L}, 2^{-n\zeta_n})$ code for $P_\s{A|X}$, we have
% \[
%     \zeta_n \geq \min_{Q_\s{XA}}\left[D(Q_\s{A|X} \| P_\s{A|X}\mid Q_\s{X}) + \max \{R -R_L- I(X;A)_Q,0\}\right] + \edited{O\pr*{\frac{\log n}{n}}},
% \]
% \edited{where the constant hidden in $O(\cdot)$ depends only on $\card{\SX}$ and $\card{\SA}$.}
% \end{lemma}
% \addtocounter{theorem}{-1}
% \endgroup

\begin{proof}
Let $(\Enc, \Dec)$ be an $(n, 2^{nR}, 2^{nR_L}, 2^{-n\zeta})$ list-decoding code, i.e., $\Enc: [2^{nR}] \to \SX^n$, $\Dec(a)$ is a subset of size $2^{nR_{L}}$ for all $a\in \SA^n$, and for all $m\in [2^{nR}]$,
\begin{align}
    P_{\s{A|X}}^{\times n}(\Dec^{-1}(m)|\Enc(m)) \geq 2^{-n\zeta}.
\end{align}
By pigeon hole principle, there exist a type $Q\in\mathcal{P}_n(\SX)$ and a subset $S$ of size $\frac{2^{nR}}{(n+1)^{|\SX|}}$ of $[2^{nR}]$ such that $\Enc(m) \in \mathcal{T}_Q$ for all $m\in S$. 
Furthermore, for any $m\in S$, we have

\begin{align}
    2^{-n\zeta} 
    &\leq P_{\s{A|X}}^{\times n}(\Dec^{-1}(m)|\Enc(m))\\
    &= \sum_{Q_{\s{A|X}}} P_{\s{A|X}}^{\times n}(\mathcal{T}_Q(\Enc(m))\cap\Dec^{-1}(m)|\Enc(m))\\
    &= \sum_{Q_{\s{A|X}}} 2^{-nD(Q_{\s{A|X}}\| P_{\s{A|X}}| Q)} \frac{|\mathcal{T}_Q(\Enc(m))\cap\Dec^{-1}(m)|}{\edited{\abs{\mathcal{T}_Q(\Enc(m))}}}.
\end{align}
\edited{
Averaging the above inequality over all $m\in S$, we obtain that
\begin{align}
    {2^{-n\zeta}}\leq  \frac{1}{|S|} \sum_{m\in S} \sum_{Q_{\s{A|X}}}2^{-nD(Q_{\s{A|X}}\| P_{\s{A|X}}| Q)}  \frac{|\mathcal{T}_Q(\Enc(m))\cap\Dec^{-1}(m)|}{\abs{\mathcal{T}_Q(\Enc(m))}}
\end{align}
Since there are at most $(n+1)^{|\SA||\SX|}$ conditional types $Q_{\s{A|X}}$, by applying the pigeon hole principle once again, we derive that there exists $Q_{\s{A|X}}$ such that
\begin{align}
    \frac{2^{-n\zeta}}{(n+1)^{|\SA||\SX|}}  \leq 2^{-nD(Q_{\s{A|X}}\| P_{\s{A|X}}| Q)}\frac{1}{|S|} \sum_{m\in S}  \frac{|\mathcal{T}_Q(\Enc(m))\cap\Dec^{-1}(m)|}{\abs{\mathcal{T}_Q(\Enc(m))}}
\end{align}
We now provide two upper bounds on the right hand side of the above inequality.} Note that 
\begin{align}
    \frac{1}{|S|} \sum_{m\in S}  \frac{|\mathcal{T}_{Q_{\s{A|X}}}(\Enc(m))\cap\Dec^{-1}(m)|}{\abs{\mathcal{T}_{Q_{\s{A|X}}}(\Enc(m))}} \leq 1
\end{align}

We also have
    \begin{align}
        \frac{1}{|S|} \sum_{m\in S}  \frac{|\mathcal{T}_{Q_{\s{A|X}}}(\Enc(m))\cap\Dec^{-1}(m)|}{\abs{\mathcal{T}_{Q_{\s{A|X}}}(\Enc(m))}} 
        &\stackrel{(a)}{\leq} \frac{(n+1)^{|\SA|}}{|S|2^{nH(Y|X)_Q}}  \sum_{m\in S}  |\mathcal{T}_{Q_{\s{A|X}}}(\Enc(m))\cap\Dec^{-1}(m)|\\
        &\stackrel{(b)}{\leq }\frac{(n+1)^{|\SA| + |\SX|}}{2^{nR}2^{nH(Y|X)_Q}}  \sum_{m\in S}  |\mathcal{T}_{Q_{\s{A|X}}}(\Enc(m))\cap\Dec^{-1}(m)|\\
        &\stackrel{(c)}{\leq }\frac{(n+1)^{|\SA| + |\SX|}}{2^{nR}2^{nH(Y|X)_Q}}  |\mathcal{T}_{Q_{\s{A}}}| 2^{nR_L}\\
        &\stackrel{(d)}{\leq } (n+1)^{|\SA| + |\SX|} 2^{n(-R+R_L-H(A|X)_Q + H(A)_Q)}\\
        &= (n+1)^{|\SA| + |\SX|} 2^{n(-R+R_L+I(A;X)_Q)}
    \end{align}
where $(a)$ follows from \eqref{eq:type3}, $(b)$ follows since $|S| \geq \frac{2^{nR}}{(n+1)^{|\SX|}}$, $(c)$ follows since $\mathcal{T}_{Q_{\s{A|X}}}(\Enc(m))\cap\Dec^{-1}(m) \subset \mathcal{T}_{Q_{\s{A}}}$ and every element of $\mathcal{T}_{Q_{\s{A}}}$ appears in $\mathcal{T}_{Q_{\s{A|X}}}(\Enc(m))\cap\Dec^{-1}(m)$ for at most $2^{nR_L}$ $m$, and $(d)$ follows from \eqref{eq:type2}.
Combining above inequalities, we obtain that
\begin{align}
    \frac{2^{-n\zeta}}{(n+1)^{|\SA||\SX|}}  \leq 2^{-nD(Q_{\s{A|X}}\| P_{\s{A|X}}| Q_{\s{X}})} \min\left(1,  (n+1)^{|\SA| + |\SX|} 2^{n(-R+R_L+I(A;X)_Q)}\right)
\end{align}
Taking $\log$ from both sides of the above inequality results in the desired bound.
\end{proof}

\section{NPA hierarchy for LSSD}\label{section Appendix NPA}

In this appendix, we describe the NPA hierarchy adapted to the LSSD setting.
For more details on the original NPA hierarchy, see \cite{WatrousNPA,NPA2008}.

Recall from \cref{sec:LSSD} that LSSD game is played by two collaborating players, Alice and Bob, who receive inputs $a \in \SA, b \in \SB$ and must produce outputs $x_A, x_B \in \SX$, respectively (see \cref{fig:LSSD_schematic}).
Here we will consider the case when the game is defined by a joint probability distribution\footnote{More generally, $P_{\s{XAB}}$ can be replaced by a quantum state $\rho_{\s{XAB}}$, see \cref{sec:quantum resources}.} $P_{\s{XAB}}$ that describes how their inputs $a,b$ are correlated with an external variable $x \in \SX$ which they need to guess.
The LSSD task is to produce outputs $x_A$ and $x_B$ such that $x_A = x_B = x$.
We can equivalently describe this by the predicate $V(x,x_A,x_B) := \delta[x_A = x_B = x]$.

Depending on the physical scenario considered, Alice and Bob might share some resource that allows them to correlate their outputs.
For the sake of generality, let
$\mathcal{C} \subset \R^{\SX \times \SX \times \SA \times \SB}$
denote an arbitrary set of correlations they can utilize.
We will treat each element $Q \in \mathcal{C}$ as a vector
$\R^{\SX \times \SX \times \SA \times \SB}$
and write its entries as $Q(x_A,x_B|a,b)$
where $x_A, x_B \in \SX, a \in \SA, b \in \SB$.
This notation emphasizes that $Q$ can also be interpreted as a stochastic matrix.
Indeed, it will always be the case that
$Q(x_A,x_B|a,b) \geq 0$ and
$\sum_{x_A,x_B \in \SX} Q(x_A,x_B|a,b) = 1$
for all $a \in \SA$ and $b \in \SB$.
For example, when dealing with quantum strategies, $Q$ has the following form, see \cref{eq:quantum value}:
\begin{align}
    Q(x_A,x_B|a,b) = \bra{\psi} \of[\big]{M_{x_A}(a) \x N_{x_B}(b)} \ket{\psi},
    \qquad
    \forall a \in \SA, b \in \SB, x_A, x_B \in \SX
    \label{eq:quantum correlations}
\end{align}
for some finite-dimensional bipartite Hilbert space $\HH = \HH_{\s{A}} \x \HH_{\s{B}}$, pure state $\ket{\psi} \in \HH$, and collections of projective measurements $\set{M_x(a) : x \in \SX} \in \PM{\HH_{\s{A}}}$ and $\set{N_x(b) : x \in \SX} \in \PM{\HH_{\s{B}}}$ on $\HH_{\s{A}}$ and $\HH_{\s{B}}$, respectively.\footnote{One can assume without loss of generality that the shared quantum state is pure and both measurements are orthogonal.}
We will denote the set of all \emph{quantum correlations} by $\mathcal{Q} \subset \R^{\SX \times \SX \times \SA \times \SB}$.

The winning probability of the LSSD game defined by a distribution $P_{\s{XAB}}$ and played with assistance of correlations $\mathcal{C}$ is given by
\begin{equation}\label{eq w general for classical LSSD}
    \omega_{\mathcal{C}} (\s{X}|\s{A};\s{B})_{P_{\s{XAB}}}
    = \sup_{Q \in \mathcal{C}}
      \sum_{\substack{x,x_A,x_B \in \SX \\ a \in \SA, b \in \SB}}
      P_{\s{XAB}}(x,a,b) V(x,x_A,x_B) Q(x_A,x_B|a,b).
\end{equation}
If we let
\begin{equation}
    K(x_A,x_B,a,b) := \sum_{x \in \SX} P_{\s{XAB}}(x,a,b) V(x,x_A,x_B),
    \label{eq:K}
\end{equation}
we can rewrite \cref{eq w general for classical LSSD} as
\begin{equation}
     \omega_{\mathcal{C}}(\s{X}|\s{A};\s{B})_{P_{\s{XAB}}}
     = \sup_{Q \in \mathcal{C}} \langle K, Q \rangle
     \label{eq:omegaC}
\end{equation}
where we treat both $K$ and $Q$ as vectors in
$\R^{\SX \times \SX \times \SA \times \SB}$
and
\begin{equation}
    \langle K, Q \rangle := \sum_{\substack{x_A, x_B \in \SX \\ a \in \SA, b \in \SB}} K(x_A,x_B,a,b) Q(x_A,x_B|a,b).
\end{equation}

To define a \emph{commuting measurement strategy}, we relax the requirement that the finite-dimensional Hilbert space $\HH$ has a tensor product structure.
We consider a pure\footnote{The assumption that the state is pure and that the measurements are projective is without loss of generality.} state $\ket{\psi} \in \HH$ and two collections of measurements on the whole of $\HH$: one measurement $\set{M_x(a) : x \in \SX}$ for each of Alice's inputs $a \in \SA$ and one measurement $\set{N_x(b) : x \in \SX}$ for each of Bob's inputs $b \in \SB$.
We require that these are orthogonal projective measurements, and that Alice's and Bob's operators pair-wise commute.
Namely, for all $a \in \SA$, $b \in \SB$, $x_A \neq x_A' \in \SX$ and $x_B \neq x_B' \in \SX$ the measurement operators should satisfy
\begin{enumerate}
    \item $M_{x_A}(a)^{\dagger}=M_{x_A}(a)$ and $N_{x_B}(b)^{\dagger}=N_{x_B}(b)$,     
    \item $M_{x_A}(a)M_{x_A'}(a)=0$ and $N_{x_B}(b)N_{{x_B'}}(b)=0$,
    \item $\sum_{x_A\in\SX}M_{x_A}(a) = \Id$ and $\sum_{x_B\in\SX}N_{x_B}(b) = \Id$,
    \item $[M_{x_A}(a),N_{x_B}(b)]=0$.
\end{enumerate}
The resulting correlation $Q$ is then defined as
\begin{equation}\label{eq def behaviour}
    Q(x_A,x_B|a,b) := \bra{\psi} M_{x_A}(a) N_{x_B}(b) \ket{\psi},
    \qquad
    \forall a \in \SA, b \in \SB, x_A, x_B \in \SX.
\end{equation}
Compared to \cref{eq:quantum correlations}, here the two commuting sets of measurements are global since the underlying space $\HH$ has no tensor product structure.

We will now construct a positive semidefinite matrix $G$ whose entries contain the values $Q(x_A,x_B|a,b)$ from \cref{eq def behaviour}, and then impose linear constraints on $G$ that capture the above conditions on the measurement operators $M_{x_A}(a)$ and $N_{x_B}(b)$.
The rows and columns of $G$ will be indexed by\footnote{We assume the sets $\SA$ and $\SB$ are disjoint so that the disjoint union makes sense here.}
\begin{equation}
    % \Sigma_1 := \set{\varepsilon} \sqcup (\SX \times \SA) \sqcup (\SX \times \SB).
    \Sigma_1 := \set{\varepsilon} \sqcup \Sigma_A \sqcup \Sigma_B
    \qquad \text{where} \qquad
    \Sigma_A := \SX \times \SA, \quad
    \Sigma_B := \SX \times \SB.
\end{equation}
For each $s \in \Sigma_1$, define a vector in $\HH$ as follows:
\begin{equation}\label{eq:psi(s)}
    \ket{\psi(s)} :=
    \begin{cases}
        \ket{\psi} & \text{if $s = \varepsilon$}, \\
        M_x(a) \ket{\psi} & \text{if $s = (x,a) \in \Sigma_A$}, \\
        N_x(b) \ket{\psi} & \text{if $s = (x,b) \in \Sigma_B$},
    \end{cases}
\end{equation}
and let $G \in \R^{\Sigma_1 \times \Sigma_1}$ be the Gram matrix of these vectors:
\begin{equation}
    G_{s,t} := \braket{\psi(s)}{\psi(t)},
    \qquad \forall s,t \in \Sigma_1.
    \label{eq:Gst}
\end{equation}
Since $G$ is a Gram matrix, it is clearly positive semidefinite:
\begin{equation}
    G \succeq 0.
\end{equation}
Notice that $G$ contains all of the values $Q(x_A,x_B|a,b)$ from \cref{eq def behaviour}, as well as some additional values such as $\braket{\psi}{\psi}$, $\bra{\psi} M_x(a) \ket{\psi}$, and others.

Because of the various relations among the measurement operators $M_{x_A}(a)$ and $N_{x_B}(b)$ listed earlier, the Gram matrix $G$ is subject to the following linear constraints:
\begin{enumerate}
    \item Since $\ket{\psi}$ is a normalized state, $\braket{\psi}{\psi} = 1$ and thus
    \begin{equation}
        G_{\varepsilon,\varepsilon} = 1.
        \label{eq:G1}
    \end{equation}
    \item Due to the completeness relations $\sum_{x \in \SX} M_x(a) = \Id = \sum_{x \in \SX} N_x(b)$,
    % for the measurement operators $M_x(a)$ and $N_x(b)$,
    we have that for any vector $\ket{v} \in \HH$,
    $\sum_{x \in \SX} \bra{\psi} M_x(a) \ket{v} = \braket{\psi}{v}$ and
    $\sum_{x \in \SX} \bra{v} M_x(a) \ket{\psi} = \braket{v}{\psi}$,
    and similarly for $N_x(b)$.
    % \begin{align}
    %     \sum_{x \in \SX} \bra{\psi} M_x(a) \ket{\psi(s)} &= \braket{\psi}{\psi(s)}, &
    %     \sum_{x \in \SX} \bra{\psi(s)} M_x(a) \ket{\psi} &= \braket{\psi(s)}{\psi}, \\ 
    %     \sum_{x \in \SX} \bra{\psi} N_x(b) \ket{\psi(s)} &= \braket{\psi}{\psi(s)}, &
    %     \sum_{x \in \SX} \bra{\psi(s)} N_x(b) \ket{\psi} &= \braket{\psi(s)}{\psi}.
    % \end{align}
    Letting $\ket{v} = \ket{\psi(s)}$ for some $s \in \Sigma_1$, this translates to
    \begin{align}
        \sum_{x \in \SX} G_{(x,a),s} &= G_{\varepsilon,s}, &
        \sum_{x \in \SX} G_{s,(x,a)} &= G_{s,\varepsilon}, &
        \forall a \in \SA, s \in \Sigma_1,
        \label{eq:G21} \\
        \sum_{x \in \SX} G_{(x,b),s} &= G_{\varepsilon,s}, &
        \sum_{x \in \SX} G_{s,(x,b)} &= G_{s,\varepsilon}, &
        \forall b \in \SB, s \in \Sigma_1.
        \label{eq:G22}
    \end{align}
    \item Since within each measurement the projectors are orthogonal, we also have $\bra{\psi} M_x(a) M_{x'}(a) \ket{\psi} = 0 = \bra{\psi} N_x(b) N_{x'}(b) \ket{\psi}$ and thus
    \begin{align}
        G_{(x,a),(x',a)} &= 0, \qquad \forall x \neq x' \in \SX, a \in \SA,
        \label{eq:G31} \\
        G_{(x,b),(x',b)} &= 0, \qquad \forall x \neq x' \in \SX, b \in \SB.
        \label{eq:G32}
    \end{align}
    \item Since $M_x(a)$ are projectors,
    $\bra{\psi} M_x(a) M_x(a) \ket{\psi} = \bra{\psi} M_x(a) \ket{\psi}$ and likewise for $N_x(b)$, so
    \begin{align}
        G_{(x,a),(x,a)} &= G_{(x,a),\varepsilon} = G_{\varepsilon,(x,a)}, \qquad \forall x \in \SX, a \in \SA,
        \label{eq:G41} \\
        G_{(x,b),(x,b)} &= G_{(x,b),\varepsilon} = G_{\varepsilon,(x,b)}, \qquad \forall x \in \SX, b \in \SB.
        \label{eq:G42}
    \end{align}
    \item Since the two sets of projectors commute,
    $\bra{\psi} M_{x_A}(a) N_{x_B}(b) \ket{\psi} =
     \bra{\psi} N_{x_B}(b) M_{x_A}(a) \ket{\psi}$,
    we have
    \begin{align}
        G_{(x_A,a),(x_B,b)} = G_{(x_B,b),(x_A,a)}, \qquad \forall x_A, x_B \in \SX, a \in \SA, b \in \SB.
        \label{eq:G5}
    \end{align}
\end{enumerate}

Let $\mathcal{Q}_1 \subset \R^{\SX \times \SX \times \SA \times \SB}$ denote the set of all correlations $Q$ such that there exists a matrix $G \in \R^{\Sigma_1 \times \Sigma_1}$ which satisfies
\begin{equation}
    G_{(x_A,a),(x_B,b)} = Q(x_A,x_B|a,b), \qquad \forall x_A,x_B \in \SX, a \in \SA, b \in \SB,
    \label{eq:G and Q}
\end{equation}
as well as $G \succeq 0$ and the linear constraints in \cref{eq:G1,eq:G21,eq:G22,eq:G31,eq:G32,eq:G41,eq:G42,eq:G5}.
Note that deciding the membership of $Q$ in $\mathcal{Q}_1$ is a semidefinite feasibility problem -- it requires finding a positive semidefinite matrix $G \succeq 0$ subject to linear constraints.

Since the local measurement operators $M_{x_A}(a) \x \Id$ and $\Id \x N_{x_B}(b)$ commute, the original set of quantum correlations $\mathcal{Q}$ defined by \cref{eq:quantum correlations} satisfies $\mathcal{Q} \subseteq \mathcal{Q}_1$.
Therefore, based on \cref{eq:omegaC},
\begin{equation}
\label{eq upperbound 1st level NPA}
    \omegaq(\s{X|A;B})_{P_{\s{XAB}}}
      := \sup_{Q \in \mathcal{Q}}   \expectedbraket{K,Q}
    \leq \sup_{Q \in \mathcal{Q}_1} \expectedbraket{K,Q}
    =: \omega_{\mathrm{q}_1}(\s{X|A;B})_{P_{\s{XAB}}},
\end{equation}
where the vector $K \in \R^{\SX \times \SX \times \SA \times \SB}$ defined in \cref{eq:K} specifies the LSSD game in question.
The value $\omega_{\mathrm{q}_1}(\s{X|A;B})_{P_{\s{XAB}}}$ corresponds to the \emph{first level} of the NPA hierarchy.
We can compute it by a semidefinite program as follows.
Define a symmetric matrix $H \in \R^{\Sigma_1 \times \Sigma_1}$ with entries
\begin{equation}
    H_{(x_A,a),(x_B,b)}
 := H_{(x_B,b),(x_A,a)}
 := \frac{1}{2} K(x_A,x_B,a,b), \qquad
  \forall x_A,x_B \in \SX, a \in \SA, b \in \SB
\end{equation}
and $0$ otherwise.
Then $\expectedbraket{K,Q} = \tr(HG)$ is a linear function of $G$, so we can compute the value of $\omega_{\mathrm{q}_1}(\s{X|A;B})_{P_{\s{XAB}}}$ via a semidefinite program that maximizes $\tr(HG)$ over all positive semidefinite matrices $G$ satisfying the conditions listed above.

The \emph{second level} of the NPA hierarchy is obtained by a similar SDP that involves a larger \emph{extended} Gram matrix $G$ whose rows and columns are indexed by\footnote{We omit $\Sigma_B \times \Sigma_A$ since Alice and Bob's operators commute.}
\begin{equation}
    \Sigma_2 := \Sigma_1
    \sqcup (\Sigma_A \times \Sigma_A)
    \sqcup (\Sigma_A \times \Sigma_B)
    \sqcup (\Sigma_B \times \Sigma_B).
\end{equation}
We extend the original set of vectors $\ket{\psi(s)}$ from \cref{eq:psi(s)} by defining new vectors for the remaining elements $s \in \Sigma_2 \setminus \Sigma_1$ as follows:
\begin{equation}
    \ket{\psi(s)} :=
    \begin{cases}
        M_{x}(a) M_{x'}(a') \ket{\psi} & \text{if $s = ((x,a),(x',a')) \in \Sigma_A \times \Sigma_A$}, \\
        M_{x}(a) N_{x'}(b') \ket{\psi} & \text{if $s = ((x,a),(x',b')) \in \Sigma_A \times \Sigma_B$}, \\
        N_{x}(b) N_{x'}(b') \ket{\psi} & \text{if $s = ((x,b),(x',b')) \in \Sigma_B \times \Sigma_B$}.
    \end{cases}
\end{equation}
As before in \cref{eq:Gst}, the entries of the extended $G$ are also given by inner products $\braket{\psi(s)}{\psi(t)}$ for all $s,t \in \Sigma_2$, and we impose additional linear constraints on them similar to those in
\cref{eq:G1,eq:G21,eq:G22,eq:G31,eq:G32,eq:G41,eq:G42,eq:G5}
to capture the fact that Alice and Bob's operators describe mutually commuting projective measurements.

We denote by $\mathcal{Q}_2 \subset \R^{\SX \times \SX \times \SA \times \SB}$ the set of all correlations $Q$ for which there exists an extended Gram matrix $G \in \R^{\Sigma_2 \times \Sigma_2}$ that agrees with $Q$ on $\Sigma_1$, see \cref{eq:G and Q}, and which satisfies the linear constraints for the second level of the NPA hierarchy.
Note that $\mathcal{Q}_2 \subseteq \mathcal{Q}_1$ since the second level imposes additional constraints compared to the first level.
Intuitively, the $\ell$-\textit{th} level of the NPA hierarchy is obtained by considering the Gram matrix of the vectors of the level $\ell-1$ plus new vectors obtained from products of $\ell$ projectors, see \cite{WatrousNPA,NPA2008} for a more formal description.

For our analysis in \cref{sec:parallelrepetitionBSCgame}, we consider the SDP for an intermediate level of the NPA hierarchy between $\mathcal{Q}_1$ and $\mathcal{Q}_2$, where $G$ is the Gram matrix for the set of vectors labelled by
\begin{equation}
    \Sigma_{1+MN} := \Sigma_1
    \sqcup (\Sigma_A \times \Sigma_B).
\end{equation}
We define $\mathcal{Q}_{1+MN}$ analogously to $\mathcal{Q}_1$ and $\mathcal{Q}_2$.
Since $\Sigma_1 \subset \Sigma_{1+MN} \subset \Sigma_2$, we have
$\mathcal{Q}_1 \supseteq \mathcal{Q}_{1+MN} \supseteq \mathcal{Q}_2 \supseteq \mathcal{Q}$
and therefore
\begin{equation}
    \sup_{Q\in\mathcal{Q}_{1+MN}}\expectedbraket{K,Q}
    =: \omega_{\mathrm{q}_{1+MN}}(\s{X|A;B})_{P_{\s{XAB}}}
    \geq \omega_{\mathrm{q}_{2}}(\s{X|A;B})_{P_{\s{XAB}}}
    \geq \omegaq(\s{X|A;B})_{P_{\s{XAB}}}.
\end{equation}

\section{Three-party binary LSSD}\label{sec:binLSSD}

In this appendix, we show (partially numerically) that there exist no probability distribution $P_\s{XABC}$, where $x,a,b$ and $c$ are all binary, such that the corresponding LSSD game can be won with higher probability using no-signalling strategies than with classical strategies. We get to this conclusion by showing that none of the no-signalling correlations at the extreme points of the no-signalling polytope can ever perform better than classical strategies. 

In the next \edited{subsection} we discuss some results on optimal classical and no-signalling strategies. These results allow us to discard some no-signalling strategies of which we know that they cannot perform better than classical strategies. For the strategies that are left, we turn to linear programming to numerically show that they also cannot perform better than classical.

\subsection{Some results on optimal strategies}

\paragraph{Multi-partite no-signalling correlations}
Up until now, we have only looked at correlations between two parties. However, the concepts of locality and no-signalling can be extended to any finite number of parties. We show how to do this extension for no-signalling correlations. 

In the case of more than two parties, a correlation is no-signalling if no subset of parties~$J$ can collectively signal to the rest of the parties $I$. So the output of the parties indexed by $I$ cannot depend on the input to the parties indexed by $J$.
\begin{definition}[Definition 11 in \cite{buhrman_et_al}]\label{def:no-signalling}
An m-partite correlation $Q_{\s{X}_1\cdots\s{X}_m|\s{A}_1\cdots\s{A}_m}$ on  $\SX_1 \times \cdots \times \SX_m\times\SA_1 \times \cdots\times \SA_m$ is called no-signalling if for any index set $I \subset \{1, \dots, m\}$ and its complement $J = \{1, \dots, m\} \setminus I$ it holds that 
\begin{align} \label{for:no-signalling_cond}
    \sum_{x_J \in \SX_J} Q(x_I, x_J \big| a_I, a_J) = \sum_{x_J \in \SX_J} Q(x_I, x_J \big| a_I, a_J') ,
\end{align}
for all $x_I \in \SX_I, a_I \in \SA_I$ and $a_J,a_J' \in \SA_J$.
\end{definition}

The next lemma states that we can loosen the constraints a little and still be left with an equivalent definition of no-signalling. Specifically, it states that it is sufficient to require that any single party cannot signal to the rest.

\begin{lemma}\label{lem:no-signal_multiple_parties}
Suppose $Q$ is a m-partite correlation satisfying \cref{for:no-signalling_cond} for all index sets $I$ such that their complements $J$ have cardinality $1$ and for all $x_I \in \SX_I, a_I \in \SA_I$ and $a_J,a_J' \in \SA_J$. Then $Q$ is a no-signalling correlation. 
\end{lemma}

\begin{proof}
We prove this lemma by induction on the cardinality of the complement $J$ of an index set $I$. If $|J| = 1$, condition \eqref{for:no-signalling_cond} holds by assumption. Now suppose $|J| = n$, and let~$x_I \in \SX_I, a_I \in \SA_I$ and $a_J,a_J' \in \SA_J$. Take $j \in J$ and let $J' = J\setminus \{j\}$. We now find 
\begin{align*}
    \sum_{x_J \in \SX_J} Q(x_I, x_J \big| a_I, a_J) &= \sum_{x_{J'} \in \SX_{J'}} \sum_{x_j \in \SX_j} Q(x_I, x_{J'}, x_j \big| a_I, a_{J'}, a_j)\\
    &\overset{(i)}{=} \sum_{x_{J'} \in \SX_{J'}} \sum_{x_j \in \SX_j} Q(x_I, x_{J'}, x_j \big| a_I, a_{J'}, a_j')\\
    &\overset{(ii)}{=} \sum_{x_{J'} \in \SX_{J'}} \sum_{x_j \in \SX_j} Q(x_I, x_{J'}, x_j \big| a_I, a_{J'}', a_j')\\
    &=\sum_{x_J \in \SX_J} Q(x_I, x_J \big| a_I, a_J'),
\end{align*}
where (i) follows by assumption on $Q$ and (ii) by induction (we are free to exchange the sums).
\end{proof}

This first lemma is an extension of the classical part of Lemma~3.2 in the paper by Majenz et al.~\cite{majenz2021local}.
It gives a list of all deterministic strategies (or more accurately: winning probability thereof) we need to consider in finding the optimal classical winning probability. The proof of this lemma relies on the relatively simple observation that the players should have equal output sets (sets consisting of all values they could possibly output according to their strategy).

\begin{lemma}\label{lem:det_strat_reduction}
Let $P_\s{XABC}$ be a probability distribution over $\SX \times \SA \times \SB \times \mathscr{C}$ with $\SA = \SB = \mathscr{C} = \{0,1\}$ and $\SX = [d]$, $d \geq 2$. The classical winning probability for $P_\s{XABC}$ is given by
\begin{align} \label{max_det}
    \omegac(\s{X}\big|\s{A};\s{B};\s{C})_P = \max_{\substack{s,t\\s\not=t}}\max \left\{ \begin{matrix}
        P_\s{X}(s), \\
        P_\s{XABC}(s,0,0,0) + P_\s{XABC}(t,1,1,1), \\
        P_\s{XABC}(s,1,0,0) + P_\s{XABC}(t,0,1,1), \\
        P_\s{XABC}(s,0,1,0) + P_\s{XABC}(t,1,0,1), \\
        P_\s{XABC}(s,0,0,1) + P_\s{XABC}(t,1,1,0)
    \end{matrix}
    \right\}.
\end{align}
\end{lemma}

\begin{proof}
First, remember that we only have to consider deterministic strategies (see \cref{sec:Classical}). Any deterministic strategy can be represented by three functions $f,g,h \colon \{0,1\} \to \SX$. Given such a strategy, the probability of winning is given by 
\begin{align}
    &\sum_{x,a,b,c} P_\s{XABC}(x,a,b,c) \delta[f(a) = g(b) = h(c) = x] = &\sum_{a,b,c} P_\s{XABC}(f(a), a,b,c) \delta[f(a) = g(b)=h(c)].\label{for:three_player_winprob}
\end{align}

Notice that there is always an optimal strategy such that $\{f(0),f(1)\} = \{g(0),g(1)\} = \{h(0),h(1)\}$. Suppose, for example, that for some $a^*$, we have that $f(a^*) \notin \{g(0), g(1)\}$. It follows that $\delta[f(a^*) = g(b) = h(c)] = 0$ for all $b,c$. Changing Alice's output on input~$a^*$, such that $f(a^*) \in \{g(0), g(1)\}$, causes $\delta[f(a^*) = g(b) = h(c)]$ to possibly be equal to $1$ for some $b,c$. This change introduces non-negative terms in the sum of \cref{for:three_player_winprob}, while not losing any others, thereby increasing the winning probability.

There are $5$ possible ways in which we have $\{f(0),f(1)\} = \{g(0),g(1)\} = \{h(0),h(1)\}$. The first is that all players ignore their input and always output some fixed $s$. In this case, the probability of winning is given by 
\begin{align*}
    \sum_{a,b,c} P_\s{XABC}(s,a,b,c) = P_X(s) \, ,
\end{align*}
yielding the first term in \cref{max_det}.
The other 4 possibilities are when they all take their input into account:
\begin{itemize}
    \item $f(0) = g(0) = h(0)$ and $f(1) = g(1) = h(1)$ or,
    \item $f(1) = g(0) = h(0)$ and $f(0) = g(1) = h(1)$ or,
    \item $f(0) = g(1) = h(0)$ and $f(1) = g(0) = h(1)$ or,
    \item $f(0) = g(0) = h(1)$ and $f(1) = g(1) = h(0)$.
\end{itemize}
defining $f(0) =: s$ and $f(1) =: t$, the winning probability in each of these cases is equal to a term in \cref{max_det}.
\end{proof}

Whereas the previous lemma reduced the number of interesting deterministic strategies, the next lemma and its corollary will do so for no-signalling strategies. 

\begin{lemma}\label{lem:ns_strat_reduction}
Let $P$ be a probability distribution over $\SX \times \SA_1 \times \cdots \times\SA_m$ with $|\SX| = d$ and $d \geq 2$. Let $Q$ be a no-signalling strategy for which 
\[
    Q(x,\dots, x \big| a_1, \dots, a_m) \leq \frac{1}{d},
\]
holds for all $x \in \SX$ and $a_1 \in \SA_1, \dots, a_m \in \SA_m$. Then its winning probability in the LSSD game defined by $P$ is at most the best classical winning probability:
\[
    \sum_{\substack{x\in \SX\\a_1 \in \SA_1, \dots, a_m \in \SA_m}} P(x,a_1, \dots, a_m)Q(x,\dots, x \big| a_1, \dots, a_m) \leq \omegac(\s{X} \big| \s{A}_1 ; \dots; \s{A}_m)_P .
\]
\end{lemma}

\begin{proof}
The proof relies on the simple fact that the $m$ players can always use deterministic strategies to win with at least probability $1/d$ by ignoring their inputs and guessing the value of $x$ to be the one most likely in $P$. The probability that the referee picks a certain value $x$ is given by $P(x) = \sum_{a\in\SA_1 \times \cdots\times\SA_m} P(x,a)$ and since $\sum_{x} P(x) = 1$, there exists an~$x^*\in\SX$ such that $P(x^*) \geq 1/d$. We conclude that $\omegac(\s{X} \big| \s{A}_1; \dots; \s{A}_m)_P \geq 1/d$.

We use the previous argument to finish the proof:
\begin{align*}
    &\sum_{\substack{x\in\SX\\ a_1 \in \SA_1, \dots, a_m \in \SA_m}} P(x,a_1,\dots,a_m)Q(x,\dots, x \big|a_1, \dots, a_m)\\ \leq& \frac{1}{d}\sum_{\substack{x\in\SX\\ a_1 \in \SA_1, \dots, a_m \in \SA_m}} P(x,a_1,\dots,a_m) = \frac{1}{d} \leq \omegac(\s{X} \big| \s{A}_1; \dots; \s{A}_m)_P.
\end{align*}
\end{proof}

\begin{corollary}\label{cor:ns_strat_reduction}
Consider an LSSD problem with $m$ players defined by a distribution $P$ for which $\omegac(\s{X}|\s{A}_1; \dots; \s{A}_m)_P < \omegans(\s{X}|\s{A}_1; \dots; \s{A}_m)_P$. There is an optimal no-signalling strategy $Q$ at one of the vertices of the no-signalling polytope, such that there exist $x \in\SX$, with $|\SX| = d$, and $a_1 \in \SA_1, \dots, a_m \in \SA_m$ for which $Q(x,\dots, x \big| a_1, \dots, a_m) > 1/d$.
\end{corollary}

\begin{proof}
Since the set of all no-signalling strategies is a convex polytope, and the winning probability of a no-signalling strategy is a linear function, we know that the optimal winning probability is achieved by a strategy $Q$ at one of the vertices of the polytope (see \cref{sec:lin_prog}). We also know that there exist $x \in \SX$ and $a_1 \in \SA_1, \dots, a_m \in \SA_m$ such that $Q(x,\dots,x | a_1, \dots, a_m) > 1/d$, because otherwise this strategy would not achieve winning probability higher than $\omegac(\s{X}|\s{A}_1; \dots; \s{A}_m)_P$ by \cref{lem:ns_strat_reduction}. 
\end{proof}

In the case of two players, we would now be done in showing that there is no binary LSSD game with a gap between no-signalling and classical winning probabilities, since all no-signalling correlations at the extreme points of the no-signalling polytope satisfy the conditions of \cref{lem:det_strat_reduction} \cite[Theorem 1]{barrett2005nonlocal}. We will see in the next section that for three players, this is not the case. However, \cref{cor:ns_strat_reduction} is still very useful as it eliminates many of the no-signalling strategies.

\subsection{No gap between classical and no-signalling}\label{sec:no_gap_three_player} 

\begin{theorem}\label{thm:nogap}
$\omegans(\s{X}|\s{A};\s{B};\s{C})_P = \omegac(\s{X}|\s{A};\s{B};\s{C})_P$ for all probability distributions $P_\s{XABC}$ over binary inputs and outputs.
\end{theorem}

\begin{proof}
Thanks to \cref{eq:pgrelations}, we can equivalently show that 
\[
    \sup_{P} \of[\big]{\omegans(\s{X}|\s{A};\s{B};\s{C})_P - \omegac(\s{X}|\s{A};\s{B};\s{C})_P} = 0.
\]
Now we have turned the problem into an optimization problem. It is, however, not possible to solve this problem using a single linear program, since the target function is not linear: the target function is the maximum of the difference between two sets. Luckily, using \cref{cor:ns_strat_reduction} and some additional tricks, we can solve this problem using multiple linear programs.

First of all, we note that the set of all probability distributions $P_\s{XABC}$ forms a convex polytope in $\R^n$. The polytope is defined by the following linear constraints:
\[
    \forall x,a,b,c \;\; P_\s{XABC}(x,a,b,c) \geq 0,
\]
and
\[
    \sum_{x,a,b,c}P_\s{XABC}(x,a,b,c) = 1.
\]

Apart from the variables that describe a probability distribution, we also add two variables $c_\text{d}$ and $c_\text{ns}$ to the linear program, which represent $\omegac(\s{X}|\s{A};\s{B};\s{C})_P$ and $\omegans(\s{X}|\s{A};\s{B};\s{C})_P$ respectively. These two variables should satisfy the following constraints:
\[
    c_\text{d} \geq \sum_{x,a,b,c} P_\s{XABC}(x,a,b,c)Q_\text{d}(x,x,x|a,b,c),
\]
for all deterministic strategies $Q_\text{d}$ and 
\begin{align}\label{for:c2_constraints}
    c_\text{ns} \geq \sum_{x,a,b,c} P_\s{XABC}(x,a,b,c)Q_\text{ns}(x,x,x|a,b,c),
\end{align}
for all no-signalling strategies $Q_\text{ns}$ at the vertices of the no-signalling polytope. 

Now, the problem is to maximize $c_\text{ns} - c_\text{d}$, which is a linear function in two variables, so we can use a linear program. However, since we have not put an upper bound on~$c_\text{ns}$, this problem is obviously unbounded. We can work around this issue by changing one of the constraints in \cref{for:c2_constraints} to an equality. Solving the linear program with one of these constraints set to an equality constraint gives us the maximum gap under the  assumption that the corresponding no-signalling strategy is the best strategy. By considering all no-signalling strategies in this way we can find the maximum gap between classical and no-signalling winning probabilities.

All that is left is to find the no-signalling strategies at the extreme points of the no-signalling polytope. We can find them using a \textit{Python} package called \textit{cddlib}, which is based on a C package under the same name \cite{fukuda}. 
Similar to linear programs, this package can provide all vertices of the polytope corresponding to a given set of linear constraints. In our case the constraints say that the strategy $Q_\text{ns}$ is a conditional probability distribution on $\SX^3 \times \SA \times \SB \times \mathscr{C}$ and it is no-signalling (where we can use \cref{lem:no-signal_multiple_parties} to omit redundant constraints). 
\edited{We find with \file{three\_player\_polytope\_extrema.py} \cite{code} that this no-signalling polytope has $53856$ extreme points, which is in line with the findings of the paper by Pironio et al.~\cite[Section 2.2]{pironio2011extremal}.}

Since the above number of extremal no-signalling strategies is quite large, we would like to reduce it so that we need to solve fewer linear programs. Using \cref{cor:ns_strat_reduction}, there must be an optimal strategy of a specific form, which reduces the number of relevant no-signalling strategies from $53\,856$ to $174$. In addition, we can also use \cref{lem:det_strat_reduction} to reduce the number of relevant deterministic strategies from $2^6=64$ to $10$.
\edited{This calculation is performed by \file{filter\_three\_player\_strategies.py} \cite{code}.}

\edited{
Now that we have everything needed to find the maximum gap between binary three-party classical and no-signalling strategies, we use the \textit{Mathematica} notebook \file{Three-party binary LSSD.nb} \cite{code} to exactly solve the above $174$ linear programs.
In each case the optimal value is $0$, meaning that there is no binary LSSD game for three players such that no-signalling resources improve its winning probability.}
\end{proof}

\end{document}